\documentclass[12pt]{article}
\usepackage{fullpage}
\usepackage{amsmath}
\usepackage{amssymb}
\usepackage{eepic}
\usepackage{epic}

\begin{document}
\bibliographystyle{plain}
\newtheorem{lemma}{Lemma}
\newtheorem{fact}{Fact}
\newtheorem{claim}{Claim}
\newtheorem{remark}{Remark}
\newtheorem{definition}{Definition}
\newtheorem{theorem}{Theorem}
\newtheorem{proposition}{Proposition}
\newtheorem{conclusion}{Conclusion}
\newtheorem{corollary}{Corollary}
\newtheorem{example}{Example}
\newtheorem{open}{Open Problem}
\newtheorem{algorithm}{Algorithm}
\newtheorem{assumption}{Assumption}

\def\squarebox#1{\hbox to #1{\hfill\vbox to #1{\vfill}}}
\newcommand{\qed}{\hspace*{\fill}
\vbox{\hrule\hbox{\vrule\squarebox{.667em}\vrule}\hrule}\smallskip}
\newenvironment{proof}{\noindent{\bf Proof:~~}}{\(\qed\)}

\author{Ahuva Mu'alem~\footnote{Computer Science Department, 
Holon Institute of Technology (HIT), Holon, Israel. ahumu@yahoo.com.} 
\and
Michael Schapira~\footnote{School of Computer Science and Engineering, Hebrew  University, Jerusalem, Israel. schapiram@huji.ac.il.} }

\title{Setting Lower Bounds on Truthfulness}

\maketitle
\date{}
\begin{abstract}
We present a general technique for proving
inapproximability results for several paradigmatic
 truthful multi-dimensional mechanism design problems. 
In particular, we demonstrate the strength of our technique by
exhibiting a lower bound of $2-\frac{1}{n}$ for the scheduling
problem with $n$ unrelated machines (formulated as a mechanism design
problem in the seminal paper of Nisan and Ronen on Algorithmic
Mechanism Design).  Our lower bound applies to  universally-truthful randomized mechanisms, regardless of any computational assumptions
 on the running time of these mechanisms. 
Moreover, it holds even for the
wider class of truthfulness-in-expectation mechanisms.  
We then turn to Bayesian mechanism design and  show a lower bound of $1.2$ for 
  Bayesian Incentive Compatible deterministic mechanisms.  
No lower bounds for truthful  mechanisms in multi-dimensional settings 
with randomness  were previously known.

We then define  the
workload-minimization problem in networks. 
We prove our lower
bounds for this problem in the inter-domain routing setting presented by Feigenbaum, Papadimitriou, Sami, and Shenker.

Finally, we discuss several notions of non-utilitarian fairness (Max-Min fairness, Min-Max fairness, and envy minimization) and show how our technique can be used to prove lower bounds for these notions.~\footnote{The current paper supersedes an earlier version  that appeared as an extended abstract in  the \emph{Proceedings of the 18th Annual ACM-SIAM Symposium on Discrete Algorithms  (SODA-07)},  pages 1143-1152,  2007. 
The current version includes a new lower bound result  
 for Bayesian Incentive Compatible Mechanisms.}
\end{abstract}

\section{Introduction}

\subsection{Inapproximability Issues in Algorithmic Mechanism Design}
\label{intro-1}
\emph{Mechanism Design} is a field of  economic theory and game-theory  that deals
with designing protocols for optimizing global goals that require
interaction with selfish players~\cite{Mas,Rubi}. \emph{Algorithmic Mechanism Design}~\cite{NR} 
combines an economic perspective that takes into account the
strategic behaviour of the players, with a theoretical
computer-science perspective that focuses on  aspects
such as computational-efficiency and approximability. 

Let us now describe, more formally, the nature of the problems
that Algorithmic Mechanism Design attempts to solve: There is a
finite set of \emph{alternatives} $A=\{a, b, c, ...\}$, and a set of
\emph{strategic players} $N=\{1, ..., n\}$. Each player $i$ has a
\emph{valuation function} $v_i: A\rightarrow \mathbb{R}$ that is his private
information. The players are self-interested and only wish to
maximize their own utility. The global goal is expressed by a
\emph{social choice function} $f$ that assigns to every possible
$n$-tuple of players' valuations $(v_1, ..., v_n)$ an alternative
$a \in A$. Mechanisms are said to \emph{truthfully implement} a
social choice function if their outcome for every $n$-tuple of
players' valuations matches that of the social choice function, and
if they enforce payments of the different players in a way that
motivates truthful report of their valuations
 (no matter what the other players do).~\footnote{It 
is well known (e.g.,~\cite{Mas}) that, 
without loss of generality, we can limit ourselves to only
considering direct-revelation truthful mechanisms. In such
mechanisms participants are always rationally  motivated to
correctly report their private information.}

A canonical social choice function is the \emph{utilitarian}
function. A utilitarian function aims to maximize the \emph{social
welfare}, i.e., to find the alternative $a\in A$ for which the
expression $\Sigma_i \ v_i(a)$ is maximized. Another paradigmatic  
 social choice function is the \emph{Max-Min} function
(based on the philosophical work of John Rawls~\cite{Rawls}). 
For every $n$-tuple of $v_i$ valuations the Max-Min function assigns the
alternative $a \in A$ that maximizes the expression $\min_i \: v_i(a)$.
Intuitively, the Max-Min function chooses the alternative $a\in A$
in which the least satisfied player has the highest value.

While in many computational and economic settings the social
choice function we wish to implement in a truthful manner is
utilitarian, often this is not the
case. Problems in which the social choice function is
non-utilitarian include revenue maximization in auctions 
(e.g.,~\cite{FGHK02}), minimizing the makespan in scheduling 
(e.g.,~\cite{NR,AT,Azar,sched-rfpas,CK}), fair allocation of resources
(e.g.,~\cite{LMNB,BV05,LMSS}), etc. A classic result of mechanism
design states that for every utilitarian problem
there exists a mechanism that truthfully implements it -- namely,
a member of the celebrated family of \emph{VCG
mechanisms}~\cite{Vic61,Clarke,Gro73}. No general technique is
known for truthfully implementing {\it non-utilitarian} social-choice
functions. In fact, some non-utilitarian social-choice functions
cannot be truthfully implemented~\cite{LMNB,NR}. 
Hence,  it is natural to ask how well  non-utilitarian 
social choice functions can be \emph{approximated} in a truthful manner.

\subsection{Our Results}

In this paper we present and discuss a general technique
for setting lower bounds on the approximability of truthful mechanisms. 
We obtain the first lower bounds for canonical  non-utilitarian settings with randomness.
 Our technique is  powerful in the following sense:
Firstly, due to its generality and simplicity it can easily be
applied to a variety of problems and notions of truthfulness (as we shall demonstrate).
Secondly, it applies to the general case of \emph{multi-dimensional 
settings}.~\footnote{As opposed to  single-dimensional settings in which the private information
of each player essentially consists of a single numerical parameter.} Finally,
it does not impose any computational assumptions on the mechanism
(such as polynomial running-time). 

In Section~\ref{sec-lower} we present our technique and
demonstrate its use on a non-utilitarian scheduling problem. The
single-dimensional version of this scheduling problem has received
much attention~\cite{AT,Azar,sched-rfpas,CK} (and references therein).
We deal with the multi-dimensional version of the problem,
formulated as a mechanism design problem by Nisan and Ronen in
their seminal paper on Algorithmic Mechanism Design~\cite{NR}:
The global goal is minimizing the makespan of the chosen schedule. I.e., assigning the tasks to
the machines in a way that minimizes the latest finishing time. 
Obviously, the makespan-minimization social choice
function is non-utilitarian and hence \emph{might} not be
truthfully implemented by any mechanism. Nisan and Ronen prove
that not only is it impossible to minimize the makespan in a
truthful manner, but that \emph{any approximation strictly better than $2$ cannot
be achieved by a truthful deterministic mechanism}.
Since a non-truthful $(1+\epsilon)$-approximation exists~\cite{HS76} (assuming constant number of machines), this raises a natural question: 

\begin{quote}
Can near-optimal
$(1+\epsilon)$-approximation {\it  truthful} mechanisms be achieved by using {\it randomization} for multi-dimensional non-utilitarian settings?
\end{quote}

Section~\ref{sec-lower} illustrates our technique by proving
several lower bounds for this problem. In particular, we prove
that no universally-truthful randomized mechanism can achieve an approximation
ratio better than $2-\frac{1}{n}$. This nearly matches the known
truthful upper bound of $1.58606$ for the case in which there
are only two  machines~\cite{Chen}.  Surprisingly, this lower bound applies even for the substantially
weaker notion of truthfulness for randomized mechanisms -
truthfulness-in-expectation.  
We also show a lower bound of $1.2$ for  Bayesian  Incentive Compatible mechanisms 
(also known as Bayesian  truthful mechanisms).
This bound applies to deterministic  
Bayesian mechanisms.
These are  the first  lower bounds  for multi-dimensional settings with randomness. In fact,
to the best of our knowledge these are the first lower bounds for
universally truthful  mechanisms, truthful-in-expectation mechanisms and Bayesian Incentive Compatible mechanisms  in \emph{multi-dimensional}  settings in general. 

\begin{quote}
Hence,  truthful $(1+\epsilon)$-approximation using randomization is  ruled out for the canonical unrelated machines problem  (regardless of computational efficiency).
\end{quote}

In addition, we show how to prove lower bounds for the important
class of \emph{strongly-monotone} deterministic mechanisms. 
The strongly-monotone property~\cite{LMN} is essentially similar to 
Arrow's Independence of Irrelevant Alternatives 
(IIA). Lavi et al.~\cite{LMN} show that in  several canonical domains this property can be assumed 
without loss of generality. This natural property  says   that the social choice 
between  two alternatives  depends only on the individual valuation difference  
between these two alternatives.~\footnote{Together with decisiveness, strong-monotonicity essentially implies affine maximization in general combinatorial auctions domains and multi-unit domains~\cite{LMN}. In several discrete domains 
(such as unrestricted integer domains),   strong-monotonicity is  sufficient for truthful implementability, 
while weak-monotonicity is not~\cite{MS08}.
For a recent characterization of strongly-monotone scheduling mechanisms see~\cite{KV-smon-15}.} 
This is another step towards proving the long-standing  conjecture of Ronen and
Nisan that \emph{no truthful deterministic mechanism can obtain an
approximation ratio better than $n$.}

In Section~\ref{sec-applications} we present  another multi-dimensional non-utilitarian problem --
minimizing the workload in communication networks. This problem
arises naturally in the design of routing mechanisms. We study the
approximability and inapproximability of this problem in the inter-domain routing
setting presented by Feigenbaum, Papadimitriou, Sami, and
Shenker~\cite{FPSS}. 

Finally, in Section~\ref{sec-fairness} we discuss three notions of
non-utilitarian fairness -- Max-Min fairness, Min-Max fairness,
and envy-minimization. We highlight the connections between these
notions and the problems studied in this paper and prove several
general inapproximability results.

\subsection{Related Work}
In a seminal paper Nisan and Ronen~\cite{NR} introduced the field of
Algorithmic Mechanism Design. The main problem presented
in~\cite{NR} to illustrate the novelty of this new area of research
was \emph{scheduling with unrelated machines}. Nisan and Ronen
explored the approximability of this non-utilitarian multi-dimensional
problem and exhibited a lower bound of $2-\epsilon$ for truthful
deterministic mechanisms. 
For this NP-hard scheduling problem  there exist an FPTAS~\cite{HS76} (assuming constant number 
of machines) 
and a polynomial-time $2$-approximation algorithm~\cite{LST87}, that 
are both non-truthful. Additionally, this problem cannot be approximated in
polynomial-time within a factor of less than $\frac{3}{2}$~\cite{LST87}.

In recent years Algorithmic Mechanism Design has been the subject
of extensive study~\cite{AGT-Book,Tim-Book-2016}. A substantial  amount of this research has focused on single-dimensional settings 
(see e.g.,~\cite{LOS,AT02,MN02,FGHK02,KKT}). 
Nearly-optimal truthful mechanisms were designed  
for \emph{the single-dimensional problem of minimum makespan for 
scheduling tasks on related machines}~\cite{AT,Azar,sched-rfpas,CK}.
The exploration of truthful mechanisms
for multi-dimensional settings has arguably mainly revolved around the
problem of welfare maximization in Multiple-Object
auctions~\cite{CA-sur,Nisan-MU-Survey}, that has gained the status of the
paradigmatic problem of this field. As this is a utilitarian
problem, it can be optimally and truthfully implemented by a VCG
mechanism. However, it has been shown that the social welfare in
combinatorial auctions cannot be maximized (or even closely
approximated) in polynomial time~\cite{LOS,NS02}. As algorithmic
mechanism design seeks time-efficient implementations, the main
challenge faced by researchers is devising truthful
polynomial-time mechanisms that approximately maximize the social
welfare in combinatorial auctions~\cite{LS05,DNS05,DNS06,BGN,HKMT,DD13}. 

There are few {\it inapproximability} results for truthful
mechanisms. This is particularly true in multi-dimensional settings.
Other than Nisan and Ronen's $2-\epsilon$ lower bound discussed
previously, the following inapproximability results are known for combinatorial auctions: Lavi,
Mu'alem and Nisan~\cite{LMN} proved that 
no polynomial time \emph{deterministic} truthful mechanism for a multi-unit auction between two players that always allocates all units can achieve an approximation factor better than $2$.
Dobzinski and Nisan~\cite{DN10} proved
inapproximability results for polynomial-time VCG mechanisms for multi-unit auctions.  
 Dobzinski and Vondr{\'{a}}k~\cite{Dob16} bound the power of polynomial-time 
universally-truthful randomized mechanisms in combinatorial auctions with submodular valuations.
Several papers use VC dimensionality  to prove inapproximability results
for deterministic truthful mechanisms for  combinatorial auctions~\cite{MIKE15} (and references therein). 
We contribute to this ongoing research by presenting methods for
deriving the first lower bounds for multi-dimensional non-utilitarian settings  that
apply to truthful mechanisms with randomness. Our
lower bounds do not require any assumptions on the running-time of
the mechanisms. 

Our technique greatly relies on the work of Bikhchandani et
al.~\cite{LMNB}. They characterize truthfulness in multi-dimensional
settings by showing that any truthful deterministic mechanism must
maintain a certain \emph{weak-monotonicity} property. Using this
characterization,~\cite{LMNB} manages to show that while welfare
maximization can be truthfully implemented in combinatorial
auctions, one cannot truthfully implement the Max-Min social choice
function, even in a very restricted type of combinatorial auctions.
The weak-monotonicity property (and several of its extensions) will
play a major role in our inapproximability proofs.

\paragraph{Makespan in multi-parameter settings.}
Nisan and Ronen~\cite{NR} present a truthful deterministic mechanism
that obtains an $n$-approximation.  They showed that no truthful
deterministic mechanism can achieve an approximation ratio strictly better than $2$ (and also strengthened this lower bound to $n$ for two specific classes 
of deterministic mechanisms whose payments satisfy some local properties)~\cite{NR}. 
This lower bound has been improved and extended in
a series of results. Christodoulou et al.~\cite{KoutsoupiasSODA2007}
 showed a lower bound of $1+\sqrt{2}$ for $3$ machines.~\footnote{Gamzu~\cite{Gamzu} gives a considerably simpler alternative  proof.} Koutsoupias and Vidali~\cite{KoutsoupiasPHI}
showed a lower bound of $1+\phi \approx 2.618$ for truthful deterministic mechanisms 
with $n\rightarrow \infty$ 
machines ($\phi$ is the golden ratio).
An optimal lower bound of $n$ 
 for {\it anonymous} truthful mechanisms is shown in~\cite{Ashlagi}. 

For the case of two machines,
Lehmann, Nisan, and Ronen exhibit a universally-truthful randomized mechanism that
obtains an approximation of $\frac{7}{4}$~\cite{NR}. This upper bound was later improved
 to $1.6737$~\cite{LY}   and then improved  to $1.58606$~\cite{Chen}.  
Dobzinski and Sundararajan show that every mechanism in the support 
of  any universally-truthful randomized mechanism for two machines that obtains finite approximation ratio must be \emph{task independent}~\cite{Dobzin-mukund-charac-08}.  
That is, the mechanism must  assign each task separately 
from the others.~\footnote{See~\cite{Vidali2009,Vidali2011} for further  characterizations of scheduling mechanisms.}

For Bayesian settings, Daskalakis and Weinberg~\cite{Costis-2015-sched} 
recently show that there is a polynomial time $2$-approximately optimal
 mechanism for makespan minimization for 
 unrelated machines.  The approximation factor in this result is with respect to the optimal Bayesian Incentive Compatible mechanism (rather than the optimal algorithm for makespan minimization) and thus is not directly comparable with the lower bound presented in our paper.  
 
 Giannakopoulos  and Kyropoulou show that the VCG mechanism achieves an 
 approximation ratio of O$(\frac{\ln n}{\ln \ln n})$ 
 when the processing  times of the tasks are independent random variables,
identical across machines~\cite{Maria-BIC-WINE15}. This essentially improves on the previously best known bound of  O($\frac{m}{n}$) given by 
 Chawla, Hartline, Malec and Sivan~\cite{Chawla-BIC-STOC13}.

For fractional settings (where the mechanism is 
allowed to split the task across several machines) Christodoulou et al.~\cite{Fractional-sched} 
present  a truthful task-independent 
$\frac{n+1}{2}$-approximation mechanism.
To compliment this result they show a lower bound of $2-\frac{1}{n}$ for any fractional truthful mechanism. They gave a lower bound of $\frac{n+1}{2}$ 
for  task-independent fractional truthful mechanisms.   
Lu~\cite{Lu2009} gives a lower bound of $1.5625$  for scale-free universally-truthful randomized mechanism for two machines, where the allocation of tasks
depends only on relative costs, not on scale.  Lavi and Swamy~\cite{Lavi-Swamy-sched} 
and Yu~\cite{Yu09} show truthful mechanisms 
in a multi-dimensional scheduling special setting where 
the processing time of a task on each machine is either 'low' or 'high'.

In several interesting settings, truthful mechanisms are essentially 
equivalent to mechanisms that select envy-free allocations 
with the smallest supporting price vectors~\cite{Gale-Demange-1985}. 
A natural question to ask is whether envy-free pricing techniques can improve the
current striking approximability and inapproximability bounds for truthful mechanisms.
Mu'alem~\cite{Mu'alem-FAIR-by-Design} observed that the optimal envy-bounds 
are far apart from the optimal truthful bounds and therefore concludes  that envy-free 
bounding techniques  cannot be applied straightforwardly to tighten the striking randomized bounds 
for minimizing the makespan on two unrelated machines.  

 In a follow-up work, Gamzu~\cite{Gamzu} 
improved our truthful lower bound for minimizing the workload in
inter-domain routing (from $\phi =\frac{1+\sqrt{5}}{2}\approx 1.618$ to $2$) and our universally-truthful randomized 
lower bound (from $\frac{3+\sqrt{5}}{4}\approx 1.309$  
to $2$).

\subsection{Open Questions}
\begin{itemize}

\item We prove lower bounds for the scheduling
problem with unrelated machines (see Section~\ref{sec-lower}) and
for the workload-minimization problem in inter-domain routing (see
Section~\ref{sec-applications}). In both problems, there are very
large gaps between the known upper and lower bounds for truthful
mechanisms (deterministic and randomized). Narrowing these gaps is
an interesting long-standing open question.

\item This paper did not make any computational assumptions on
mechanisms. Proving (possibly stronger) lower bounds for
\emph{polynomial time} truthful mechanisms is a big open question.

\end{itemize}

\subsection{The Organization of the Paper}
In Section~\ref{sec-lower} we study several lower bounds on truthfulness of 
\emph{scheduling problem with unrelated machines}.  
In Section~\ref{sec-applications} we study  the problem of \emph{workload-minimization in networks}. 
In Section~\ref{sec-fairness} we study several
notions of non-utilitarian fairness.

\section{Preliminaries} 
\label{Section-Model}

We consider the standard  mechanism design setting: 
There are $n$ players, and a finite set of alternatives $A$.
 Each player $i \in [n]$ has a private valuation function $v_i \in V_i$ 
 that assigns a non-negative real value to every $a \in A$
(the higher the value of the alternative the more desirable it is). 

 A (deterministic) mechanism $M(f, p)$ consists of a  deterministic social choice function  
$f : V \rightarrow A$ (representing some global goal, e.g., makespan as defined below), and a payment function  $p_i : V\rightarrow R$ for each player $i$.   
Each player $i$ is simultaneously being asked to report a valuation  $v_i$ (possibly deviating from his
private valuation), and the mechanism then computes the outcome $f(v)$ and charges price $p_i(v)$ 
to player $i$ (notice that  payment might be negative or positive).  
The (quasi-linear) {\it utility} that player $i$ derives by declaring valuation  $v'_i$ is
$v_i(f(v'_i, v_{-i})) - p_i(v'_i, v_{-i})$  (assuming player $i$'s  private valuation function is $v_i$). 
Each player aims to maximize his own  utility.  
 
A mechanism $M(f, p)$ is called truthful if  for every player $i$, for every
$v_{-i} \in  V_{-i}$, and for every $v_i, v'_i  \in V_i$,

\[
v_i(f(v_i, v_{-i}))  - p_i(v_i, v_{-i}) \ge  v_i(f(v'_i , v_{-i}))  - p_i(v'_i, v_{-i}).
\]

That is, a mechanism is truthful if no player can ever improve its utility
by misreporting his private valuation to the mechanism (no matter what the other 
players do). 

\begin{remark} The above setting refers to $\mathrm{value \;  scenarios}$  where players report their  valuations to the mechanism and  willing to maximize  their values  minus their  payment  to  the mechanism.  
In  what follows we will mainly consider  $\mathrm{cost  \;  scenarios}$ where players report their costs (rather than valuations)  to the mechanism and analogously are willing to minimize  their costs minus the payment made to them by the mechanism. We will slightly abuse notation by letting $v_i$ refer 
both to a valuation function, and to a cost function, but the meaning will be clear from the context. 
\end{remark}

The approximation ratio of a mechanism optimizing 
 a global minimization goal (e.g., makespan) 
 is defined to be the worst case ratio (over  all possible valuations  $v \in V$)  between
the  goal value of the chosen alternative  and the optimal value.

A {\it universally-truthful randomized mechanism}  is a
probability distribution over truthful mechanisms. 
Formally, for every $v \in V$ the universally-truthful randomized mechanism produces a distribution $D(v)$ 
over deterministic truthful mechanisms and outputs a deterministic mechanism drawn from this distribution.

A randomized social choice function $f$ is a function from $n$-tuples of players'
valuations to probability distributions over the set of
alternatives $A$. 
A randomized mechanism  consists of a  randomized social choice function   
$f : V \rightarrow A$, and a payment function  $p_i : V\rightarrow R$ for each player $i$.
A  randomized mechanism $(f, p)$ is {\it truthful-in-expectation} if for every player $i$, for every
$v_{-i} \in  V_{-i}$, and for every $v_i, v'_i  \in V_i$,

\[
E[v_i(f(v_i, v_{-i}))  - p_i(v_i, v_{-i})] \ge  E[v_i(f(v'_i , v_{-i}))  - p_i(v'_i, v_{-i})].
\]
Thus, if a mechanism is truthful-in-expectation then the \emph{expected} utility of a player is maximized when he declares his true private valuation $v_i$ (no matter what the other players do). 
Here, we assume {\it risk neutral} players aiming  to maximize the expected difference between their  true private valuation and their total payment (the expectation is taken over
any randomness in the mechanism).

The approximation ratio of a mechanism with randomness  optimizing a global minimization goal (e.g., makespan)  is defined to be the worst case ratio (over  all possible valuations $v \in V$)  between
the {\it expected} goal value of the chosen alternative  and the optimal value. 

\vspace{0.5cm}
In a Bayesian (a.k.a. stochastic) setting  the
private valuations of each player $i$  is drawn independently from $D_i$. 
The product distribution 
$D = D_1 \times  \cdots  \times D_n$  is assumed to be public knowledge. 
We restrict our attention to 
{\it deterministic}  mechanisms (with deterministic social choice functions).  
A mechanism $M(f, p)$ is Bayesian Incentive  Compatible 
 (given the public knowledge distribution $D$) if  for all $i$
\[
E_{v_{-i} \sim D_{-i}}  \left[ v _i(f(v_i,v_{-i})) - p_i(v_i, v_{-i}) \right] \ge 
E_{v_{-i} \sim D_{-i}}  \left[ v_i(f(v'_i, v_{-i})) - p_i(v'_i, v_{-i})\right]. 
\]

That is, player $i$'s  expected utility from reporting his true
valuation $v_i$  is no less than his expected utility from reporting 
a different valuation $v'_i$ when
others' true valuations are drawn from the product distribution $D_{-i}$.
Here, we  assume {\it risk neutral} players aiming  to maximize the expected difference between their  true private valuation and their total payment (the expectation is taken over
the randomness in other players' valuations).
Thus, if a  mechanism is Bayesian Incentive  Compatible  then the \emph{expected} utility of a player is maximized when he declares his true private valuation $v_i$ (assuming all other
players truthfully report their valuations).  

The approximation ratio of a Bayesian mechanism  optimizing a global minimization goal (e.g., makespan)   is defined to be the ratio  between
the {\it expected} goal value of the chosen alternative  
and the {\it expected} optimal value, where the expectations are taken over $D$. 

\section{A Presentation of Our Technique Via the Scheduling Problem} 
\label{sec-lower}
In this section we present our technique and show how it can be used to derive lower bounds
for the scheduling problem with $n$ unrelated machines. Nisan and
Ronen~\cite{NR} exhibited a truthful $n$-approximation deterministic
mechanism for this problem. Their mechanism is a VCG
mechanism, and can easily be shown to be strongly-monotone (see
Subsection~\ref{sub-deterministic} 
for a formal definition of strong-monotonicity). 
This mechanism can be viewed as auctioning each task separately  in 
a Vickrey auction.~\footnote{Christodoulou et al. show that for two machines VCG is the unique mechanism  achieving the optimal approximation of $2$~\cite{CKV-carac-08}.} It is straight forward to show that any VCG mechanism with a \emph{deterministic} tie-breaking rule among alternatives is strongly-monotone. They also proved a lower bound of $2-\epsilon$ for
truthful deterministic mechanisms that applies even when there are
only two machines and is tight for this case. 
However, Nisan and Ronen~\cite{NR}
conjecture that their lower bound is not tight in general, and that
\emph{any truthful deterministic mechanism cannot obtain an approximation
ratio better than $n$}.

For the case of two machines,  Nisan and Ronen show that randomness
helps get an approximation ratio better than $2$. They present a 
universally-truthful randomized mechanism that has an approximation ratio of
$\frac{7}{4}$. We generalize their result by designing a 
universally-truthful randomized mechanism that obtains an approximation ratio of
$\frac{7n}{8}$ (see Appendix~\ref{upper-bound}).\emph{Thus, 
we prove that randomness achieves better
performances than the known truthful deterministic $n$ upper bound
for any number of machines}.
In Subsection~\ref{sub-deterministic} we show ways of proving lower
bounds for truthful deterministic mechanisms. Using these methods we
provide a simple and shorter proof for Nisan and Ronen's
$2-\epsilon$ lower bound. Our proof (unlike the original) relies on
exploiting the \emph{weak-monotonicity} property defined
in~\cite{LMNB}. Subsection~\ref{sub-deterministic}
also aid us in deriving a stronger lower bound for the important
classes of \emph{strongly-monotone} deterministic mechanisms. 
We note, that the mechanism in~\cite{NR}, which is the best
currently known deterministic mechanism for the scheduling problem,
is contained in this class. We prove that no approximation ratio
better than $n$ is possible for this class of strongly-monotone deterministic mechanisms ({\it thus
making another step towards proving the long-standing conjecture of~\cite{NR}}).

After discussing lower bounds for truthful deterministic mechanisms
we turn our attention  mechanisms with randomness. There are
two possible definitions for the truthfulness of a randomized
mechanism~\cite{DNS06,NR}. The first and stronger one is that of
\emph{universal truthfulness} that defines a truthful randomized
mechanism as a probability distribution over truthful deterministic
mechanisms. Thus, this definition requires that for \emph{any} toss of the random coins 
 made by the mechanism, the players
still maximize their utility by reporting their true valuations. A
considerably weaker definition of truthfulness is that of
\emph{truthfulness-in-expectation}. This definition only requires
that players maximize their \emph{expected} utility, where the
expectation is over the random choices of the mechanism (but still
for every behaviour of the other players). Unlike universally
truthful mechanisms,
 truthful-in-expectation mechanisms  only motivate {\em risk-neutral} bidders to act
truthfully. Risk-averse bidders may benefit from strategic behaviour.
In addition, truthful-in-expectation mechanisms induce
truthful behaviour only as long as players have no information about
the outcomes of the random coin flips before they need to act.

In Subsection~\ref{sub-universal} we prove the first lower bound
on the approximability of universally-truthful randomized mechanisms in
multi-dimensional settings. Namely, we show that any universally-truthful randomized mechanism for the scheduling problem cannot achieve an
approximation ratio better than $2-\frac{1}{n}$. This lower bound
nearly matches the universally truthful  $1.58606$
upper bound for the case of two machines~\cite{Chen}. To prove this lower
bound, we make use of a general technique that is based on Yao's
powerful principle~\cite{Yao}. Our proof for the $2-\epsilon$ lower bound for
deterministic mechanisms (in Subsection~\ref{sub-deterministic})
serves as a building block in the proof of this lower
bound.

In Subsection~\ref{sub-expectation} we strengthen this result by
proving that the same lower bound holds even when one is willing
to settle for truthfulness-in-expectation. Our proof relies on
some of the ideas that appear in the proof of the previous lower
bound but takes a different approach. In particular, we generalize
the weak-monotonicity requirement to fit the class of truthful
randomized mechanisms, and explore the implications of this
extended monotonicity on the probability distributions over
allocations generated by such mechanisms.

In Subsection~\ref{sub-sec-LB-Bayes}
we turn to the notion of Bayesian Incentive Compatible mechanisms, where players' valuations are drawn from a distribution that is public knowledge, and show a lower bound of $1.2$  for deterministic Bayesian mechanisms.

\subsection{The Setting}
\label{SubSection-Sched-Setting}
We consider the standard  unrelated strategic machines setting~\cite{NR}:
There are $m$ tasks that are to be assigned on $n$ machines, 
where each task must be assigned to exactly one machine. 
Every machine $i$ is a strategic player with an arbitrary 
valuation function $v_i:2^{[m]}\rightarrow \mathbb{R_{+}}$ such that
$v_i(\{j\})$ (or $v_i(j)$, for short)  denotes the private \emph{cost}  of task $j\in[m]$ on machine $i$. 

One can think of the private cost of task $j$ on machine $i$ as the
time it takes $i$ to complete $j$. For every subset of tasks  $S\subseteq [m]$,
$v_i(S)=\Sigma_{j\in S} \ v_i(j)$. 
That is, the total cost of a set of tasks on machine $i$ is the  
sum of the costs of the individual tasks on that machine. 
Each machine wishes to minimize the cost of the 
set of tasks  assigned to it minus the payment made to it by the mechanism.
If the mechanism is truthful then machine $i$ can never improve its utility
by misreporting its private cost function $v_i$ to the mechanism (no matter what the other 
machines do). 

The set of alternatives $A$ contains all  possible  allocations of tasks to the machines, where
all tasks must be assigned, and each task is assigned to exactly one machine. 
The global goal is minimizing the makespan of the chosen allocation.
 I.e., find an allocation of tasks $a \in A$ to minimize the expression 
  \[
 \max\{v_1(a_1), \ldots, v_n(a_n)\}.
  \]

\subsection{Lower Bounds for Truthful Deterministic Mechanisms}
\label{sub-deterministic}
Bikhchandani et al.~\cite{LMNB} formally 
define the weak-monotonicity property for deterministic mechanisms: Consider an
Algorithmic Mechanism Design setting with $n$ strategic players. Before we present the
formal definition we will require the following

\begin{definition}\label{def-weak-monotonicity}
Let $M$ be a  deterministic mechanism. Let $i\in [n]$ and
let $v=(v_1, ..., v_n)$ be an $n$-tuple of players' valuations. Let
$v'_i$ be a valuation function. Denote by $a$ the alternative that $M$
outputs for $v$ and by $b$ the alternative that $M$ outputs for
$(v'_i,v_{-i})$. The mechanism $M$ is said to be $\mathrm{weakly \; monotone}$ if for all such
$i$, $v$, and $v'_i$ it holds that: 
\[
v_i(a)+v'_i(b)\ge v'_i(a)+v_i(b).
\]
\end{definition}

\begin{remark}
This definition of weak-monotonicity is for value scenarios  in which each
player wishes to maximize the difference between his valuation 
and his  total payment. 
In cost scenarios in which players have costs (such as the scheduling of unrelated machines problem, the workload
minimization problem, and the min-max fairness  considered in this paper) the inequality is
in the other direction.
\end{remark}

Bikhchandani et al.~\cite{LMNB} prove that any truthful deterministic mechanism must
be weakly-monotone. For completeness, we present this simple proof.

\begin{lemma}\label{monotonicity-lemma}
Any truthful deterministic mechanism must be weakly-monotone.
\end{lemma}
\begin{proof}
Let $M$ be a truthful deterministic mechanism. Let $i\in [n]$ and
let $v=(v_1, ..., v_n)$ be an $n$-tuple of players' valuations. Let
$v'_i$ be a valuation function. Denote by $a$ the alternative that $M$
outputs for $v$ and by $b$ the alternative that $M$ outputs for
$(v'_i,v_{-i})$. Consider player $i$. It is well known that the
price a player is charged by the mechanism to ensure his
truthfulness cannot depend on the player's report. Specifically, the
payment of player $i$ in $a$ and $b$ is a function of $v_{-i}$ and
of $a$ and $b$, respectively. We denote by $p_i(v_{-i},a)$ and by
$p_i(v_{-i},b)$ $i$'s payment in $a$ and $b$, respectively. It must
hold that $v_i(a)-p_i(v_{-i},a)\ge v_i(b)-p_i(v_{-i},b)$ (for
otherwise, if $i$'s valuation function is $v_i$, he would have an
incentive to declare his valuation to be $v'_i$). Similarly,
$v'_i(b)-p_i(v_{-i},b)\ge v'_i(a)-p_i(v_{-i},a)$. By adding these
two inequalities we reach the weak-monotonicity requirement.
\end{proof}

\vspace{0.1in}

Relying on the weak-monotonicity property we provide an
alternative proof for the $2-\epsilon$ lower bound of~\cite{NR}
for the scheduling problem with unrelated machines. Our proof
shows that any deterministic mechanism that achieves an
approximation ratio better than $2$ violates the weak-monotonicity
property.

\begin{theorem}\label{deterministic-bound}
Any weakly-monotone mechanism cannot achieve an approximation
ratio better than $2$.
\end{theorem}

\begin{proof}
Let  $\epsilon$ be an arbitrarily small positive real number. 
Consider the scheduling problem with two machines
and three tasks. For every machine $i=1, 2$ we define two possible
valuation functions $v_i$ and $v'_i$:

\[
v_i(t)=\left\{%
\begin{array}{ll}
1 & t=i\ or\ t=3\\
100 & \mbox{otherwise} \\
\end{array}%
\right.
\]

\[
v'_i(t)=\left\{%
\begin{array}{ll}
0 & t=i \\
1+\epsilon & t=3 \\
100 & \mbox{otherwise.} \\
\end{array}%
\right.
\]

Let $M$ be a deterministic, weakly-monotone, mechanism that
achieves an approximation ratio better than $2$. Then, when
players $1$ and $2$ have the valuations $v_1$ and $v_2$
respectively, $M$ must assign task $1$ to player $1$, task $2$ to
player $2$, and can choose to which player to assign task $3$
(because the optimal makespan is $2$ and any other assignment
results in a makespan of at least $100$). W.l.o.g. assume that $M$
assigns task $3$ to player $2$. Now, consider the instance with
players' valuations $(v'_1, v_2)$. Notice that the only
task-allocation that guarantees an approximation ratio better than
$2$ is assigning tasks $1$ and $2$ to players $1$ and $2$
respectively, and task $3$ to player $1$. However, this turns out
to be a violation of the weak-monotonicity requirement. 
Weak-monotonicity, in this case, dictates that for every player $i=1, 2$
it must hold $v_i(a)+v'_i(b)\le v'_i(a)+v_i(b)$. However, if we
look at player $1$ we find that
$1+(1+\epsilon)=v_1(1)+v'_1(\{1, 2\})>v'_1(1)+v_1(\{1, 2\})=2$. A
contradiction.
\end{proof}

\vspace{0.1in}

Lavi et al.~\cite{LMN} present and study another property -- strong-monotonicity. 
They show that in some canonical domains this property can be assumed 
without loss of generality.  Strong-monotonicity is the strict version of weak-monotonicity (Definition \ref{def-weak-monotonicity}). It says essentially  that the social choice between  two alternatives  depend only on the individual valuation difference  
between these two alternatives.

\begin{definition}
Let $M$ be a deterministic mechanism. Let $i\in [n]$ and let
$v=(v_1, ..., v_n)$ be an $n$-tuple of players' valuations. Let
$v'_i$ be a valuation function. Denote by $a$ the alternative that $M$
outputs for $v$ and by $b$ the alternative that $M$ outputs for
$(v'_i,v_{-i})$. The mechanism $M$ is said to be $\mathrm{strongly \; monotone}$ if for all
such $i$, $v$, and $v'_i$ it holds that: If $a\neq b$, then
\[
v_i(a)+v'_i(b)>v'_i(a)+v_i(b).
\]
\end{definition}

Here, too,  the inequality is
in the other direction if players have costs rather than values. 
We prove that no member of the class of strongly-monotone mechanisms
can obtain an approximation better than $n$ for the scheduling
problem (even for the case of zero/one valuations). 
The idea at the heart of our proof of
Theorem~\ref{SMON-bound} is an iterative use of 
the strong-monotonicity property to construct an instance of the problem for
which the allocation generated by the mechanisms is very far
from optimal.

\begin{theorem}\label{SMON-bound}
Any strongly-monotone mechanism 
cannot obtain an approximation ratio better than $n$.
\end{theorem}

\begin{proof}
Consider an instance of the scheduling problem with $n$ machines
and $m=n^2$ tasks. Let $M$ be a deterministic mechanism for which
the strong-monotonicity property holds. Let $\mathcal{I}$ be the instance of
the scheduling problem in which every machine $i$ has a valuation
function $v_i$ such that $v_i(j)=1$ for all $j\in[m]$. Denote by
$S=(S_1, ..., S_n)$ the allocation of tasks produced by $M$ for the 
instance $\mathcal{I}$. It must be that there is some machine $r$ such that
$|S_r|\ge n$. Without loss of generality let $r=n$.

We will now create a new instance $\mathcal{I}^1$ by altering the valuation
function of machine $1$ to $v'_1$ while leaving all the other
valuation functions unchanged (in case $S_1=\emptyset$ we skip
this part). That is, machine $1$ will have the valuation function
$v'_1$:

\[
v'_1(t)=\left\{%
\begin{array}{ll}
0 & t\in S_1 \\
1 & t\notin S_1 \\
\end{array}%
\right.
\]
and every other machine $i\neq 1$ will have a valuation function $v_i$. 
Denote by $T=(T_1, ..., T_n)$ the allocation $M$
generates for $\mathcal{I}^1$. The first step of the proof is showing that
$S_1=T_1$. This is guaranteed by the strong-monotonicity of $M$.
Assume, by contradiction that $S_1\neq T_1$. The strong-monotonicity
property ensures that $v_1(S_1)+v'_1(T_1) <  v_1(T_1)+v'_1(S_1)$.
By assigning values we have:
\[
|S_1|+|T_1\setminus S_1|< |T_1| + 0. 
\]
Observe, that $|S_1|+|T_1\setminus S_1|-|T_1|=|S_1\setminus T_1|$,
therefore:
\[
|S_1\setminus T_1|<0.
\]
A contradiction.

We shall now prove that not only does $S_1$ equal $T_1$, but in
fact $S_i = T_i$ for every $i$. 
Since $S_1 = T_1$ it must be
that $v_1(S_1)+v'_1(T_1)=v'_1(S_1)+v_1(T_1)$.
 However, the strong-monotonicity property 
 (with respect to $(v_1, v_{-1})$ and $(v'_1, v_{-1})$) 
dictates that if this is true then $S=T$.

In an analogous manner we shall now turn the valuation function of
machine $2$ into $v'_2$ while keeping all the other valuation
functions in $\mathcal{I}^1$ unchanged (in case $S_2=\emptyset$ we skip this
part). That is, the valuation function of machine $2$ is changed into:

\[
v'_2(t)=\left\{%
\begin{array}{ll}
0 & t\in S_2 \\
1 & t\notin S_2. \\
\end{array}%
\right.
\]

Similar arguments show that the allocation produced by the $M$ for
this new instance will remain $S$. We can now iteratively continue
to change the valuation functions of machines $3, ..., n-1$ into
$v'_3, ..., v'_{n-1}$ respectively, without changing the allocation
the mechanism generates for these new instances. After performing
this, we are left with an instance in which every machine
$i \in [n-1]$ has the valuation function $v'_i$, and machine $n$ has
the valuation function $v_n$. We have shown that the allocation
generated by $M$ for this instance is $S$. Recall  that
$|S_n|\ge n$. Fix some  $R\subseteq S_n$ such that $|R|=n$. We will now
create a new instance $\mathcal{I}^n$ from the previous one by only altering
the valuation function $v_n$ into the following valuation function
$v'_n$:

\[
v'_n(t)=\left\{%
\begin{array}{ll}
0 & t\in S_n\setminus R \\
1 & \hbox{otherwise.} \\
\end{array}%
\right.
\]

By applying similar arguments to the ones used before, one can
show that the allocation generated by $M$ when given the instance
$\mathcal{I}^n$ remains $S$. Observe that the finishing time of $S$ for $\mathcal{I}^n$ is $n$
because all the tasks in $R$ are assigned to machine $n$. Also
notice that the finishing time of the optimal allocation of
tasks for $\mathcal{I}^n$ is precisely $1$ (by assigning the $i$'th task in $R$ to machine $i$,  for $i=1, ..., n$). The theorem follows.
\end{proof}

\subsection{A Lower Bound for Universally-Truthful Randomized Mechanisms}
\label{sub-universal}

We now present a technique for deriving lower bounds for
universally truthful mechanisms, based on Yao's principle~\cite{Yao}.
 Consider a zero-sum game with two players. Let the
row player's strategies be the various different instances of
a specific problem, and let the
column player's strategies be all the deterministic truthful
mechanisms for solving that problem. Let entry $g_{ij}$ in the
matrix $G$ depicting the game be the approximation ratio obtained
by the algorithm of column $j$ when given the instance of row $i$.

Recall that every universally-truthful randomized mechanism
 is a probability distribution over deterministic
truthful mechanisms. The straight-forward approach for proving a lower
bound for such mechanisms  is to find an instance of the
problem on which every such mechanism cannot achieve (in
expectation) a certain approximation factor. By applying the well
known Minimax Theorem to the game described above we get that an
alternate and just as powerful way for setting lower bounds is to
show that there is a probability distribution over instances on
which any deterministic mechanism cannot obtain (in expectation) a
certain approximation ratio.

We demonstrate this technique by proving a $2-\frac{1}{n}$ lower
bound for universally truthful mechanisms for the scheduling
problem. Our proof is based on finding a probability distribution
over instances of the scheduling problem for which no
deterministic truthful mechanism can provide an approximation
ratio better than $2-\frac{1}{n}$. To show this, we shall exploit
the weak-monotonicity property of truthful deterministic
mechanisms (as discussed in Subsection~\ref{sub-deterministic}).

\begin{theorem}\label{universal-bound}
Any universally truthful mechanism 
cannot achieve an approximation ratio better than $2-\frac{1}{n}$.
\end{theorem}

\begin{proof}
Let  $\epsilon$ be an arbitrarily small positive real number. Consider the scheduling problem
with $n$ machines and $m=n+1$ tasks. For every machine $i$ we
define two possible valuation functions:

\[
v_i(t)=\left\{%
\begin{array}{ll}
1 & t=i\ or\ t=n+1\\
\frac{4}{\epsilon} & \hbox{otherwise} \\
\end{array}%
\right.
\]

and

\[
v'_i(t)=\left\{%
\begin{array}{ll}
0 & t=i \\
1+\epsilon & t=n+1 \\
\frac{4}{\epsilon} & \hbox{otherwise.} \\
\end{array}%
\right.
\]

Let $\mathcal{I}$ be the instance in which the valuation function of every
machine $i$ is $v_i$. For every $j$, let $\mathcal{I}^j$ be the
instance in which every machine $i\neq j$ has the valuation
function $v_i$, and machine $j$ has the valuation function $v'_j$.
We are now ready to define the probability distribution $P$ over
instances: instance $\mathcal{I}$ is assigned probability $\epsilon$,
and for every $j$ instance $\mathcal{I}^j$ is picked with probability
$\frac{1-\epsilon}{n}$.

We now need to show that any deterministic truthful mechanism $M$
cannot achieve an approximation ratio better than $2-\frac{1}{n}$
on $P$. Let $T^j$ be the allocation of the $n+1$ tasks to the $n$
machines in which every machine $i$ gets task $i$, and machine $j$
is also assigned task $n+1$. Observe, that $T^j$ is the optimal
allocation of tasks for the instance $\mathcal{I}^j$. Also observe, that while
the finishing time of the allocation $T^j$ for the instance $\mathcal{I}^j$ is
$1+\epsilon$, the finishing time of any other allocation of tasks
is at least $2$. We shall denote the allocation that $M$ outputs for
the instance $\mathcal{I}$ by $M(\mathcal{I})$. 
Similarly, we shall denote the allocation that $M$ outputs
 for the instance $\mathcal{I}^j$ by $M(\mathcal{I}^j)$ (for every $j\in[n]$).
We will now examine two distinct cases: The case in which
$M(\mathcal{I})\neq T^r$ for any $r\in [n]$, and the case that $M(\mathcal{I})=T^r$
for some $r\in [n]$.

Observe, that in the first case the finishing time is at
least $\frac{4}{\epsilon}$ while the optimal finishing time is $2$.
 Since instance $\mathcal{I}$ appears in $P$ with probability $\epsilon$ we have that
$M$'s expected approximation ratio is at least
\[
\frac{\frac{4}{\epsilon}\times \epsilon + (1+\epsilon)\times (1-\epsilon)}
{2\times \epsilon + (1+\epsilon)\times (1-\epsilon)} > 2-\frac{1}{n}.
\]

We are left with the case in which $M(\mathcal{I})=T^r$ for some $r\in [n]$.
Consider an instance $\mathcal{I}^j$ such that $j\neq r$. The following lemma
states that $M$ will not output the optimal allocation for $\mathcal{I}^j$
(that is, $T^j$).

\begin{lemma} \label{monotone1}
If $M(\mathcal{I})=T^r$ for some $r\in [n]$, then for every $j\neq r$
$M(\mathcal{I}^j)\neq T^j$.
\end{lemma} 

\begin{proof}
Let $j\neq r$. Let us assume by contradiction that $M(\mathcal{I}^j)=T^j$.
The weak-monotonicity property dictates that
$v_j({j})+v'_j(\{j,n+1\})\le v'_j({j})+v_j(\{j,n+1\})$. By
assigning values we have that $1+(1+\epsilon)\le (1+1)$, 
and reach a contradiction.
\end{proof}

\vspace{0.1in}

From Lemma~\ref{monotone1} we learn that if $M(\mathcal{I})=T^r$ (for some
$r\in [n]$) then we have that  $M(\mathcal{I}^j)$ (for every $j\neq r$) is an
allocation that is not the optimal one (i.e., not $T^j$). In fact
(as mentioned before), any allocation that  $M$ outputs given $\mathcal{I}^j$ will
have a finishing time of at least $2$, while the optimal
allocation ($T^j$) has a finishing time of $1+\epsilon$.  The expected approximation
ratio of $M$ for $P$ is therefore at least

\[ \frac{2 \cdot \frac{(n-1)\cdot(1-\epsilon)}{n} + (1+\epsilon)\cdot \frac{1-\epsilon}{n} + 2\epsilon}{(1+\epsilon)\cdot (1-\epsilon)+ 2\epsilon} \ge 
\frac{(2-\frac{1}{n})\cdot(1-\epsilon)}{(1+\epsilon)\cdot (1-\epsilon)+ 2\epsilon}
\]

Since $\lim_{\epsilon \to 0} \frac{(2-\frac{1}{n})\cdot(1-\epsilon)} {(1+\epsilon)\cdot (1-\epsilon)+ 2\epsilon} = 2-\frac{1}{n}$,  the theorem follows. 
\end{proof}

\subsection{A Lower Bound for Mechanisms that are Truthful-in-Expectation} 
\label{sub-expectation}
After handling the case of universally truthful mechanisms we now
turn to the weaker notion of truthfulness-in-expectation. Any such mechanism can be regarded as
a mechanism that for every instance of a problem produces a
probability distribution over possible alternatives. In this subsection 
we consider  risk neutral players with quasi-linear utility functions.

We start by generalizing the weak-monotonicity definition. 
We consider the expected value with respect to a distribution 
over the set of alternatives $A$: 

\begin{definition}\label{def-in-expectation-valuation-function}
Let $v$ be a valuation function. We define the 
extended valuation function $V_{v}$ as follows. For every probability
distribution $P$ over the set of alternatives $A$
\[
V_{v}(P)=\Sigma_{a\in A}\ Pr_ P[a]\times v(a).
\]
\end{definition}

Arguments similar to those of Lemma~\ref{monotonicity-lemma} show
that randomized mechanisms that are truthful-in-expectations must
be weakly monotone (given the new definition of the valuation
functions). This \emph{extended weak-monotonicity} is equivalent
to the \emph{monotonicity-in-expectation} property defined by Lavi
and Swamy~\cite{LS05}.

\begin{definition}
Let $M$ be a randomized mechanism. Let $i\in [n]$ and let
$v=(v_1, ..., v_n)$ be an $n$-tuple of players' valuations. Let
$v'_i$ be a valuation function. Denote by $P$ the distribution
over alternatives that $M$ outputs for $v$ and by $Q$ the distribution
over alternatives that $M$ outputs for $(v'_i, v_{-i})$. 
The mechanism  $M$ is said to
be $\mathrm{weakly \; monotone\; in\; the \; extended \; sense}$
 if for all such $i$, $v$,
and $v'_i$ it holds that:
 $V_{v_i}(P)+V_{v'_i}(Q)\ge
V_{v'_i}(P)+V_{v_i}(Q)$.
\end{definition}

\begin{lemma}[Lavi and Swamy~\cite{LS05}]\label{rand-monotonicity-lemma}
Any truthful-in-expectation mechanism must be weakly-monotone in the extended sense.
\end{lemma}

We can exploit this extended definition of weak-monotonicity to
prove inapproximability results. We show how this is done by
strengthening our $2-\frac{1}{n}$ lower bound for universally-truthful randomized mechanisms by showing that it applies even for the case
of truthfulness-in-expectation. 

A key element in the proof of Theorem~\ref{expectation-bound} is
the observation that instead of regarding a randomized mechanism
for the scheduling problem as generating probability distributions
over allocations of tasks, it can be regarded as generating, for
each task, a probability distribution over the machines it is
assigned to by the mechanism.  
This different view of a
randomized mechanism for this specific problem, enables us to
better analyse the contribution of each task to the expected makespan.

The main lemma in the proof of Theorem~\ref{expectation-bound},
namely Lemma~\ref{monotone2}, makes use of this fact together with
the extended weak-monotonicity condition.
Lemma~\ref{monotone2} essentially proves that for two carefully
chosen instances of the problem, the probability that a specific
task is assigned to a specific machine in one of the
instances, cannot be considerably higher than the probability it
is assigned to the same machine in the other. Thus, we show that
even though allocating this task to that machine in one of the
instances leads to a good approximation, any truthful-in-expectation
mechanism will fail to do so.

\begin{theorem}\label{expectation-bound}
Any  mechanism that is weakly-monotone in the extended
sense  cannot achieve an approximation ratio better than
$2-\frac{1}{n}$.
\end{theorem}

\begin{proof}
Let  $\epsilon$ be an arbitrarily small positive real number. 
 Consider the scheduling problem
with $n$ machines and $m=n+1$ tasks. For every machine $i$
we define two possible valuation functions:

\begin{equation}
v_i(t)=\left\{%
\begin{array}{ll}
1 & t=i\ or\  t=n+1\\
\frac{4}{\epsilon^2} & \hbox{otherwise} \\
\end{array}%
\right.
\label{eqn:v} 
\end{equation}

and
\begin{equation}
v'_i(t)=\left\{%
\begin{array}{ll}
0 & t=i \\
1+\epsilon & t=n+1 \\
\frac{4}{\epsilon^2} & \hbox{otherwise.} \\
\end{array}%
\right.
\label{eqn:v'} 
\end{equation}

Let $\mathcal{I}$ be the instance in which the valuation function of every
machine $i$ is $v_i$. For every $j\in [n]$ let $\mathcal{I}^j$ be the
instance in which every machine $i\neq j$ has the valuation
function $v_i$, and machine $j$ has the valuation function $v'_j$.
Let $T^j$ be the allocation of the $n+1$ tasks to the $n$ machines
in which every machine $i$ gets task $i$, and machine $j$ is also
assigned task $n+1$.

Let $M$ be a mechanism that is weakly-monotone in the
extended sense. We shall denote by $P$ the distribution over all
possible allocations produced by $M$ when given instance $\mathcal{I}$, and
by $P^j$ the distribution over all possible allocations $M$
produces when given instance $\mathcal{I}^j$. Let $R$ be some distribution
over the possible allocations. Fix a machine $i$ and a task $t$,
we define $p_{i,t}(R)$ to be the probability that machine $i$ gets
task $t$ given $R$. Formally,
$p_{i,t}(R)=\Sigma_{a|t\in a_i}\ Pr_R[a]$. Observe that
$V_{v_i}(R)=\Sigma_{t\in[m]}\ p_{i, t}(R)v_i(t)$ and
$V_{v'_i}(R)=\Sigma_{t\in[m]}\ p_{i, t}(R)v'_i(t)$.

We are now ready to prove the theorem. In order to do so, we prove
that for every mechanism $M$, as defined above, one can find an
instance of the scheduling problem for which $M$ fails to give an
approximation ratio better than $2-\frac{1}{n}$. 
Consider instance $\mathcal{I}$.

If  $p_{i, i}(P)<1-\epsilon^2$ (for some $i\in[n]$) then machine $i$
does not get task $i$ with probability of at least $\epsilon^2$.
However, when machine $i$ does not get task $i$, the finishing
time of a schedule for $\mathcal{I}$ cannot be less than
$\frac{4}{\epsilon^2}$, while the optimal finish time is $2$.
Therefore, with probability of at least $\epsilon^2$ the
makespan of the  algorithm is at least
$\frac{4}{\epsilon^2}$. If this is the case, 
the approximation ratio is at least $2$ (and the theorem follows).
Hence, from now on we will only deal with the case in which for
every $i\in[n]$, 
\[
p_{i, i}(P)\ge 1-\epsilon^2.
\]

Let $r$ be some machine such that $p_{r, n+1}(P)\le \frac{1}{n}$ 
(recall that $\Sigma_{i \in [n]}\  p_{i, t}(P) = 1$ for every task $t \in [m]).$
Intuitively, $r$ is a machine that is hardly assigned task $n+1$
in $P$. 

We will show that in this case we can choose the instance
$\mathcal{I}^{r}$ to prove our lower bound. The main idea of the proof is
showing that machine $r$ will not be assigned task $n+1$ in $P^r$
with probability that is significantly higher than the probability
it was assigned the task in $P$. 

\begin{lemma}\label{monotone2}
Let $r$ be some machine such that $p_{r, n+1}(P)\le \frac{1}{n}$.
It holds that $p_{r, n+1}(P^r)\le \frac{1}{n}+\epsilon$.
\end{lemma}
\begin{proof}
As $M$ is weakly-monotone in the extended sense (and since this is a cost scenario) we have that
$V_{v_r}(P)+V_{v'_r}(P^r) \le V_{v'_r}(P)+V_{v_r}(P^r)$. That is:
\[
\;\; \Sigma_{t\in[m]} \ p_{r, t}(P)v_r(t)+\Sigma_{t\in[m]} \ p_{r, t}(P^r)v'_r(t)\le
\]
\[
\Sigma_{t\in[m]} \ p_{r, t}(P)v'_r(t)+\Sigma_{t\in[m]} \ p_{r, t}(P^r)v_r(t).
\]
After subtracting identical terms from both sides 
we have:
\[
\;\;p_{r, r}(P)v_r(r)+p_{r, n+1}(P)v_r(n+1)+p_{r, r}(P^r)v'_r(r)+p_{r, n+1}(P^r)v'_r(n+1)\le
\]
\[
p_{r, r}(P)v'_r(r)+p_{r, n+1}(P)v'_r(n+1)+p_{r, r}(P^r)v_r(r)+p_{r, n+1}(P^r)v_r(n+1)
\]
By~(\ref{eqn:v}) and~(\ref{eqn:v'}) we have:
\[
p_{r, r}(P)+p_{r, n+1}(P)+p_{r, n+1}(P^r)\times (1+\epsilon)\le
p_{r, n+1}(P)\times (1+\epsilon)+p_{r, r}(P^r)+p_{r, n+1}(P^r)
\]
Therefore:
\[
p_{r, r}(P)+p_{r, n+1}(P^r)\times\epsilon\le
p_{r, n+1}(P)\times\epsilon+p_{r, r}(P^r)
\]
Because $p_{r, r}(P)\ge 1-\epsilon^2$ and $p_{r, r}(P^r)\le 1$ we have:
\[
(1-\epsilon^2)+p_{r, n+1}(P^r)\times \epsilon\le
p_{r, n+1}(P)\times \epsilon+1
\]
Equivalently, $p_{r, n+1}(P^r)\le p_{r, n+1}(P)+\epsilon.$
\end{proof}

\vspace{0.1in}

We next show that $M$ fails to provide an
approximation ratio better than $2-\frac{1}{n}$ for $\mathcal{I}^{r}$. The
optimal allocation for $\mathcal{I}^{ r}$ is $T^r$, which has a finishing time
of $1+\epsilon$. Any other allocation has a finishing time of at
least $2$. However,  with high
probability $T^r$ is not reached by $M$, since $p_{r,n+1}(P^r)\le
\frac{1}{n}+\epsilon$, and since machine $r$ gets task $n+1$ in
$T^r$, we know the probability that $M$ outputs $T^r$ is at most
$\frac{1}{n}+\epsilon$. 
The expected approximation ratio of $M$ is
therefore at least
\[
\frac {(\frac{1}{n}+\epsilon)\times(1+\epsilon)  +(1-(\frac{1}{n}+\epsilon))\times2}{1+\epsilon} \ge 
(2-\frac{1}{n})\cdot \frac {1-\epsilon}{1+\epsilon}
\] 
Since $\lim_{\epsilon \to 0} \frac {1-\epsilon}{1+\epsilon} = 1$,  the theorem follows.
\end{proof}

\subsection{A Lower Bound for Bayesian Incentive Compatible Mechanisms} 
\label{sub-sec-LB-Bayes}
We now turn to the notion of Bayesian mechanisms, where 
players' valuations are drawn from a  distribution that is public knowledge. 
We show a lower bound of $1.2$ (even for two machines  and three tasks). 
No lower bound for Bayesian Incentive Compatible mechanisms in multi-parameter settings was previously known.
In what follows,  we restrict our attention to 
deterministic mechanisms that deterministically output an allocation (in particular, for any given input each task will be always deterministically allocated to the same machine).

\begin{theorem}\label{bayesian-bound}
Any   Bayesian Incentive Compatible deterministic mechanism  
cannot achieve an approximation ratio strictly better than $1.2$. 
\end{theorem}

\begin{proof}
Let  $\epsilon$ be an arbitrarily small positive real number. 
Consider a setting with two machines and three tasks
(the generalization for $n > 2$  is straightforward).
 We define a product distribution with  
two possible  \emph{equally likely}
 valuation functions $v_i$ and $v'_i$ for every machine $i=1, 2$ 
 (notice that the processing times of the tasks are not identical across machines):

\[
v_i(t)=\left\{%
\begin{array}{ll}
1 & t=i\ or\ t=3\\
\frac{4}{\epsilon} & \hbox{otherwise} \\
\end{array}%
\right.
\]

and

\[
v'_i(t)=\left\{%
\begin{array}{ll}
0 & t=i \\
1+\epsilon & t=3\\
\frac{4}{\epsilon} & \hbox{otherwise.} \\
\end{array}%
\right.
\]

Assume by contradiction that there exists a 
\emph{deterministic} Bayesian Incentive Compatible mechanism $M$
with an expected approximation ratio $1.2 -\delta$  for some $\delta > 0$. 
We shall denote the allocation that $M$ outputs for
the instance $\mathcal{I}$ by $M(\mathcal{I})$. 

Let $T^j$ be the allocation of the tasks to machines
in which task $i$ is assigned to machine $i$, and task $3$ is assigned to machine $j$
   (where $i, j \in \{1,\ 2\}$).  
   Notice that  $M(v_1, v_2) , M(v'_1, v'_2) \in \{T^1, T^2\}$ since for any allocation that gives the first task to the second machine or the second task to the first  machine we have 
\[
\frac{\frac 1 4 \cdot \frac{4}{\epsilon}+\frac 3 4 \cdot(1+\epsilon)}{ \frac 1 4 \cdot 2 + \frac 3 4\cdot(1+\epsilon)}   > 1.2 - \delta.
\]

   Now, $T^1, T^2$  are optimal 
   for the instances $(v'_1, v_2)$, $(v_1, v'_2)$, respectively.  Furthermore, we must have that 
$M(v'_1, v_2)=T^1, M(v_1, v'_2)=T^2$  under the assumption that an expected approximation ratio strictly better than $1.2$ can be achieved. 
More formally, observe that the finishing time of any other allocation of tasks
is at least $2$, and clearly  for a small
enough $\epsilon$  we have
\[
\frac{\frac 2 4 \cdot 2  +\frac 2 4 \cdot (1+\epsilon)}{ \frac 1 4 \cdot 2 + \frac 3 4 \cdot(1+\epsilon) }   > 1.2 - \delta.
\]

Without loss of generality assume that $M(v_1, v_2) = T^2$.  
By symmetry it is enough to consider two cases: $M(v'_1, v'_2)=T^1$ and 
$M(v'_1, v'_2)=T^2$.

In the first case we have that $M(v_1, v_2)=M(v_1, v'_2) = T^2, M(v'_1, v_2)=M(v'_1, v'_2)=T^1$. 
Recall that the probability that the valuation function of machine 2 is $v_2$ (or $v'_2$)
is $\frac12$. 
By Bayesian incentive compatibility (with respect to player 1)  
we have that (for the case where the valuation function of player 1  is $v_1$):
\[
   \frac12 (p_1(v_1, v_2) - v_1(T^2)) + \frac12 (p_1(v_1, v'_2) - v_1(T^2))  \geq
   \frac12 (p_1(v'_1, v_2) - v_1(T^1)) + \frac12 (p_1(v'_1, v'_2) - v_1(T^1)).
 \]

By Bayesian incentive compatibility (with respect to player 1)  
we have that (for the case where the valuation function of machine 1  is $v'_1$):
\[
\frac12 (p_1(v'_1, v_2) - v'_1(T^1)) + \frac12 (p_1(v'_1, v'_2) - v'_1(T^1))  \geq
 \frac12 (p_1(v_1, v_2) - v'_1(T^2)) + \frac12 (p_1(v_1, v'_2) - v'_1(T^2)).
\]     
  
  By adding up these two inequalities and  cancelling out  the payments from both sides we have that  
  $v_1(T^1) + v'_1(T^2) \ge v_1(T^2) + v'_1(T^1).$  
  A contradiction (since $2 + 0 < 1 + (1+\epsilon)$). 
  
 We are left with the case in which $M(v_1, v_2)=M(v_1, v'_2)=M(v'_1, v'_2)=T^2,  \; M(v'_1, v_2)=T^1$.
By Bayesian incentive compatibility we have that  
\[
 \frac12 (p_1(v_1, v_2) - v_1(T^2)) + \frac12 (p_1(v_1, v'_2) - v_1(T^2))  \geq
   \frac12 (p_1(v'_1, v_2) - v_1(T^1)) +  \frac12 (p_1(v'_1, v'_2) - v_1(T^2)),
 \]
and
\[
\frac12 (p_1(v'_1, v_2) - v'_1(T^1)) + \frac12 (p_1(v'_1, v'_2) - v'_1(T^2))  \geq
   \frac12 (p_1(v_1, v_2) - v'_1(T^2)) + \frac12 (p_1(v_1, v'_2) - v'_1(T^2)).
\]   
Equivalently, 
  \[
 \frac12 (p_1(v_1, v_2) - v_1(T^2)) + \frac12 p_1(v_1, v'_2)  \geq
   \frac12 (p_1(v'_1, v_2) - v_1(T^1)) +  \frac12 p_1(v'_1, v'_2) ,
 \]
and
\[
\frac12 (p_1(v'_1, v_2) - v'_1(T^1)) + \frac12 p_1(v'_1, v'_2)  \geq
   \frac12 (p_1(v_1, v_2) - v'_1(T^2)) + \frac12 p_1(v_1, v'_2). 
\]   
 
 By adding up these two inequalities and rearranging we have that 
  $v_1(T^1) + v'_1(T^2) \ge v_1(T^2) + v'_1(T^1).$ 
  A contradiction (since $2 + 0 < 1 + (1+\epsilon)$). 
\end{proof}

\section{Workload Minimization in Inter-Domain Routing}\label{sec-applications}
In this section, we study another
 non-utilitarian multi-dimensional problem -- \emph{workload
minimization in inter-domain routing}. Feigenbaum, Papadimitriou,
Sami, and Shenker formulated the inter-domain routing problem as a
distributed mechanism design problem~\cite{FPSS} (inspired by the
extensive literature on the real-life problem of inter-domain
routing in the Internet). Several works that study
their model and its extensions have been published~\cite{FKMS,FRS,FSS}. 
All these works deal with the
realization of utilitarian social-choice functions, and focus on the efficient and
distributed design of VCG mechanisms.

Workload minimization is a problem that arises naturally in the
design of routing protocols, as we wish that 
no single Autonomous System (AS) will be overloaded with work.
 It can easily be shown
that any  VCG mechanism performs very poorly with respect  to
workload minimization. Thus, while optimally minimizing the total
cost, or maximizing the social welfare, the known truthful
mechanisms for this problem can result in workloads that are very
far from optimal (in which one AS is burdened by the traffic sent
by all other ASes). We initiate the study of truthful workload
minimization in inter-domain routing by presenting constant lower
bounds that apply to any truthful mechanism (deterministic and randomized).

\subsection{The Setting}

We are given a \emph{directed graph} $G=\langle N, L\rangle$ (called the \emph{AS
graph}) in which the set of nodes $N$ corresponds to the
Autonomous Systems (ASes) of which the Internet is comprised. The set $N$
consists of a \emph{destination} node $d$, and $n$ \emph{source} nodes 
 (see Example~\ref{example:WL} below). 
The set of edges $L$ corresponds to communication links
between the ASes. Each source node $i\in[n]$ is a strategic player. The
number of packets (intensity of traffic) originating in source
node $i$ and destined for $d$ is denoted by $t_i$. 
 The directed graph $G$ and number of packets $t_i$s are public knowledge. 

Let $neighbours(i)$ be all the ASes that are directly  linked to $i$ in the AS
graph. Each source node $i$ has a private cost function
$c_i\mbox{: }neighbours(i)\rightarrow \mathbb{R_{+}}$ that specifies the
per-packet cost incurred by this node for carrying traffic.  This
cost function represents the additional  load imposed on the
internal AS network when sending a packet from $i$ to an adjacent
AS. In the formulation of the problem in~\cite{FPSS}, a node
does not incur a cost for packets that originate in that node.
However, since we are interested in  workload minimization, this is not
the case in our formulation. 
Additionally, as we are interested in proving lower bounds we can restrict our attention 
to the model in which the number of packets $t_i$s are public  knowledge.
If the mechanism is truthful then player $i$ can never improve its utility
by misreporting his private cost function $c_i$ to the mechanism (no matter what the other 
players do). 
In the \emph{single-dimensional version of
this problem} an AS $i$ incurs the \emph{same} per-packet private  cost $c_i$ for
sending traffic to each of its neighbours. 

The set of alternatives $A$ contains all possible \emph{route allocation} that 
form a confluent tree to the
destination $d$. I.e., no source node is allowed to transfer traffic to two
adjacent nodes. We seek truthful mechanisms that output routing
trees in which the {\it workload}  imposed on the busiest source node is
minimized. Formally, let $N^T_i$ be the set of all nodes whose paths
in the routing tree $T$ go through node $i$. Let $s^T(i)$ be the subsequent
node $i$ transfers traffic to in the routing tree   $T$. 
The global goal is to minimize the
expression 
\[
max_i\ \Sigma_{j\in N^T_i}\ t_j\times c_i(s^T(i))
\]
over all possible routing trees $T$. Each source node $i$ is a strategic  player who
wishes to minimize his workload  ($=\Sigma_{j\in N^T_i}\ t_j\times c_i(s^T(i))$) minus the payment made to him by the mechanism. 
If the mechanism is truthful then player $i$ can never improve its utility
by misreporting his private cost  function $c_i$ to the mechanism (no matter what the other players do).

\begin{remark}
It is easy to verify that the single-dimensional related machine 
scheduling makespan minimization problem~\cite{AT}
 is a special important case of the of the single-dimensional workload
minimization problem (furthermore, from any $c$-approximation truthful mechanism  for workload minimization one can construct a $c$-approximation truthful mechanism  for makespan minimization). 
However, the problem of multi-dimensional workload
minimization is {\it strategically} different from the multi-dimensional unrelated machine scheduling makespan minimization problem that we study in Section~\ref{sec-lower} 
(e.g., since $t_i$ is a public information and the cost $c_i$ is a function of 
the direct successor nodes  rather than the direct predecessor nodes).   
\end{remark}

\begin{figure}[hbt]
\begin{center}
\setlength{\unitlength}{0.00052493in}
\begingroup\makeatletter\ifx\SetFigFont\undefined%
\gdef\SetFigFont#1#2#3#4#5{%
  \reset@font\fontsize{#1}{#2pt}%
  \fontfamily{#3}\fontseries{#4}\fontshape{#5}%
  \selectfont}%
\fi\endgroup%
{\renewcommand{\dashlinestretch}{30}
\begin{picture}(2133,2550)(0,-10)
\thicklines
\path(1813,2223)(328,1323)
\blacken\path(481.158,1446.518)(328.000,1323.000)(508.369,1401.620)(481.158,1446.518)
\path(1799,317)(314,1217)
\blacken\path(494.369,1138.380)(314.000,1217.000)(467.158,1093.482)(494.369,1138.380)
\path(1770,1275)(420,1275)
\blacken\path(615.000,1305.000)(420.000,1275.000)(615.000,1245.000)(615.000,1305.000)
\path(1815,2179)(1815,1365)
\blacken\path(1785.000,1560.000)(1815.000,1365.000)(1845.000,1560.000)(1785.000,1560.000)
\path(1815,371)(1815,1185)
\blacken\path(1845.000,990.000)(1815.000,1185.000)(1785.000,990.000)(1845.000,990.000)
\put(1995,1185){\makebox(0,0)[lb]{\smash{{\SetFigFont{12}{14.4}{\rmdefault}{\mddefault}{\updefault}x}}}}
\put(1275,1320){\makebox(0,0)[lb]{\smash{{\SetFigFont{10}{12.0}{\rmdefault}{\mddefault}{\updefault}$1$}}}}
\put(375,1815){\makebox(0,0)[lb]{\smash{{\SetFigFont{10}{12.0}{\rmdefault}{\mddefault}{\updefault}$1+\epsilon$}}}}
\put(1860,1815){\makebox(0,0)[lb]{\smash{{\SetFigFont{10}{12.0}{\rmdefault}{\mddefault}{\updefault}$0$}}}}
\put(1860,555){\makebox(0,0)[lb]{\smash{{\SetFigFont{10}{12.0}{\rmdefault}{\mddefault}{\updefault}$0$}}}}
\put(375,555){\makebox(0,0)[lb]{\smash{{\SetFigFont{10}{12.0}{\rmdefault}{\mddefault}{\updefault}$1+\epsilon$}}}}
\put(1680,2265){\makebox(0,0)[lb]{\smash{{\SetFigFont{12}{14.4}{\rmdefault}{\mddefault}{\updefault}y}}}}
\put(1635,15){\makebox(0,0)[lb]{\smash{{\SetFigFont{12}{14.4}{\rmdefault}{\mddefault}{\updefault}z}}}}
\put(15,1185){\makebox(0,0)[lb]{\smash{{\SetFigFont{12}{14.4}{\rmdefault}{\mddefault}{\updefault}d}}}}
\end{picture}
} 
\caption{} \label{fig1}
\end{center}
\end{figure}

\begin{example}\label{example:WL}
Consider the routing instance in Figure $1$. Each source node has a
single packet it wishes to send to the destination. The number
beside every directed link $(u, v)$ in the figure represents the
cost $u$ incurs for transferring a packet to $v$. 
Consider the routing tree in which both $y$ and $z$
send packets through $x$, and $x$ forwards packets directly to $d$.
The workload on $x$ is $3$ in this routing tree. In the optimal routing tree all source nodes
chose to send their packets directly to $d$. This routing tree has a
maximal workload of $1+\epsilon$. 
\end{example}

\subsection{Approximability of the Single-Dimensional Case}

It is easy to show (via a simple reduction from Partition) that
even the single-dimensional version of the workload-minimization
problem is NP-hard. However, is it at all possible to \emph{optimally}
solve this problem in a truthful manner?  The answer to this
question is yes. However, the worst-case running time is exponential 
 (we note that this is also the case in the
single-dimensional version of machine-scheduling~\cite{AT}).

\begin{lemma}
There exists a truthful, deterministic, exponential-time mechanism
that always finds a workload-minimizing route allocation in the
single-dimensional case.
\end{lemma}

\begin{proof}
The mechanism $M$ simply goes over all possible route allocations
and outputs the {\it optimal} one with respect  to workload-minimization.
As in~\cite{AT}, our truthful mechanism outputs the
lexicographically-minimal optimal route allocation. Specifically, let
$a$ and $b$ be two distinct optimal route allocations (if {\it two}  such
allocations exist). 
Let $a_1, ..., a_n$ be a \emph{decreasing} order of the
the number of packets  that go through each of different nodes in $a$ 
(using a deterministic  tie-breaking  rule). 
Similarly, let $b_1, ..., b_n$ be \emph{decreasing} order of the 
the number of packets  that go through each of different nodes in $b$. 
Let $j \in [n]$ be the first index such
that $a_j\neq b_j$ or $j=n$ if no such index exists (notice that the sorted order of the coordinates of $a$ might be  identical to the sorted order of the coordinates of $b$,  even if $a \neq b$).
The mechanism will choose $a$ if $a_j<b_j$,
$b$ if $b_j<a_j$, and otherwise according to a
predefined deterministic tie-breaking rule.

We next show the truthfulness of the mechanism. It is well
known that a mechanism is truthful in a single-dimensional setting
such as ours \emph{if and only if} it is weakly-monotone~\cite{AT}. 
Let $a$ be the route allocation that $M$ outputs when the per-packet cost of $i$ is $c_i$, and
the per-packet costs of the other nodes are
$c_{-i}=c_1, ..., c_{i-1}, c_{i+1}, ..., c_n$. 
Let $b$ be the route allocation that $M$ outputs when the per-packet cost of $i$ is $c'_i$,
and the per-packet costs of the other nodes are $c_{-i}$. 
Weak-monotonicity states that if $c_i < c'_i$ then $k_i\ge k'_i$, where
$k_i$ and $k'_i$ are the number of packets that go through $i$ in $a$
and $b$, respectively (and this is true for every node $i$, for every
vector of costs per-packet $c_{-i}$ of the other nodes, and for
every two costs per-packet $c_i\neq c'_i$).

Fix a node $i$. Assume, by contradiction, that there are
 $c_i < c'_i$, and $c_{-i}$ such that $k_i < k'_i$.
 Let $a_1, ..., a_n$ and $b_1, ..., b_n$ be defined as before. 
 Let $c=(c_i, c_{-i}), \ c'=(c'_i, c_{-i}).$
 We shall use the  notation 
 $\mbox{W}(a, c)$ to denote the maximum workload of a player in the route allocation $a$ when the per-packet costs are $c=(c_1, c_2, ..., c_n)$ 
 (that is, $\mbox{W}(a, c)=max_i\ \Sigma_{j\in N^T_i}\ t_j\times c_i(s^a(i))$). 
 Similarly,  $\mbox{W}(b, c)$ denotes the maximum workload of a player in the route allocation $b$ when the  per-packet costs are $c=(c_1, c_2, ..., c_n)$.
The terms  $\mbox{W}(a, c')$ and $\mbox{W}(b, c')$ are defined analogously.  
By the optimality of $M$ it is immediate to verify that:

\begin{equation}
\mbox{W}(a, c) \le \mbox{W}(b, c) \le \mbox{W}(b, c') \le \mbox{W}(a, c')
\label{eqn:WL1} 
\end{equation}

\begin{equation}
 \mbox{W}(a, c') = \max\{ \mbox{W}(a, c),\ k_i c'_i\}
\label{eqn:WL2} 
\end{equation}

\begin{equation}
 k'_i c'_i \le \mbox{W}(b, c') 
\label{eqn:WL3} 
\end{equation}

By~(\ref{eqn:WL2})   it suffices to consider two cases: 

\noindent{\it Case 1:} $\mbox{W}(a, c') = \mbox{W}(a, c)$. 
From~(\ref{eqn:WL1}) we have 
\[
\mbox{W}(a, c) =\mbox{W}(b, c) =\mbox{W}(b, c') =\mbox{W}(a, c').
\] 
 Now, $\mbox{W}(a, c) =\mbox{W}(b, c)$ dictates that 
 $a$ comes before $b$ in the lexicographic order (since $a$ is the  (lexicographically) optimal route allocation at $c$),
 and $\mbox{W}(b, c') =\mbox{W}(a, c')$ dictates that 
 $b$ comes before $a$ in the lexicographic order (since $b$ is the (lexicographically) optimal route allocation at $c'$), a contradiction. 
\vspace{0.12cm} 

\noindent{\it Case 2:} $\mbox{W}(a, c') = k_i c'_i$. 
From~(\ref{eqn:WL1}) and~(\ref{eqn:WL3}) 
we have  $k'_i c'_i \le \mbox{W}(b, c') \le \mbox{W}(a, c') = k_i c'_i$, contradicting 
our assumption that $k'_i > k_i$. 
\end{proof}

\subsection{Approximability of the Multi-Dimensional Case}

Feigenbaum et al.~\cite{FPSS} present a truthful polynomial-time VCG mechanism that
always outputs the \emph{cost-minimizing tree} (a tree that minimizes the
total sum of costs incurred for the packets sent  to $d$). 
We begin our discussion on the multi-dimensional version of the workload
minimization problem by showing that this VCG mechanism obtains an
$n$-approximation for the multi-dimensional version of our problem
(and hence also for the single-dimensional version) in polynomial time.

\begin{theorem}
There is a truthful polynomial-time deterministic
$n$-approximation mechanism for the workload minimization problem
in inter-domain routing.
\end{theorem}
\begin{proof}
We prove that any mechanism that minimizes the total cost provides
an $n$-approximation to the minimal workload. Hence, the mechanism
of~\cite{FPSS} obtains the required approximation ratio.

Suppose the per packet costs are $c=(c_1, ..., c_n)$. Denote by $T$ the corresponding cost-minimizing routing-tree and by $T^*$ the corresponding workload-minimizing routing-tree. Let $C(T, c)$ and $C(T^*, c)$ be the
total costs of $T$ and $T^*$, respectively. 
Recall that $W(T^*, c)$ is the value of the optimal solution 
for the workload-minimization problem. 
The result now follows
immediately from 
$C(T^*, c) \le n \cdot W(T^*, c) $ and 
 $C(T, c) \ge  W(T^*, c)$.
\end{proof}

\vspace{0.1in}

Unfortunately, it can be shown that any mechanism that minimizes the
total cost (and in particular the mechanism in~\cite{FPSS}) 
cannot obtain a better  approximation ratio.

\begin{theorem}
Any mechanism that minimizes the total cost of the routing tree
cannot achieve an approximation ratio strictly better than $n$ for the
workload minimization problem in inter-domain routing.
\end{theorem}
\begin{proof}
Recall the routing instance in Figure $1$. Observe, that any total-cost minimizing
mechanism would choose the routing tree in which both $y$ and $z$
send packets through $x$, and $x$ forwards packets directly to $d$.
This means that the workload on $x$ is $3$. However, if all nodes
chose to send their packets directly to $d$ we would reach a
maximum workload of $1+\epsilon$. 
Clearly, the example in Figure $1$ can be
generalized to $n$ source nodes.
Notice also that a similar example can be used to show that 
the theorem  holds for singe-dimensional problems. The idea is to replace the link  $(y, d)$ 
with the links $(y, y'), (y', d)$ and the link $(z, d)$ with $(z, z'), (z', d)$
while  $c_y(y, y')=c_z(z, z')=0$, and $c_{y'}=c_{z'}=1+\epsilon$ 
assuming $x, y, z$ have  each a single packet to send, and $y', z'$ 
 have no packets  to send to the destination.
\end{proof}

\vspace{0.1in}

Therefore, there exists a trade-off between the goal of minimizing the
total-cost and the goal of minimizing the workload. It would be
interesting to construct a truthful mechanism that optimizes (or at
least closely approximates) the minimal workload. 
We present two
negative results for this problem in  Appendix~\ref{LBs-Inter-Domain}. 
In a follow-up work, Gamzu~\cite{Gamzu} 
improved our truthful lower bound for minimizing the workload in
inter-domain routing from $\frac{1 + \sqrt{5}}{2} \approx 1.618$ to $2$, and our universally-truthful randomized  lower bound from $\frac{3 + \sqrt{5}}{4} \approx 1.309$ to $2$.

\section{Non-Utilitarian Fairness}\label{sec-fairness}

Utilitarian functions represent the {\it overall} satisfaction  of the
players, as they maximize the sum of players' values. This notion of
\emph{fairness} is but one of several that have been considered
(explicitly and implicitly) in mathematical, economic and
computational literature. A well known example of \emph{non-utilitarian}
fairness is the \emph{cake-cutting problem}, presented by the Polish school
of mathematicians in the 1950's 
(Steinhaus, Banach, Knaster~\cite{Stein}). Fair allocations of
indivisible items  have also
been studied~\cite{LMN,LMSS} (these can be regarded as discrete
versions of the cake-cutting problem).

In this section, we discuss three general notions of {\it non-utilitarian}
fairness -- Max-Min fairness, Min-Max fairness, and
envy-minimization. We show that Max-Min fairness is inapproximable within \emph{any
ratio}, even for extremely restricted special cases. In sharp contrast, we
show that Min-Max fairness (which is a generalization of both the
scheduling and workload-minimization problems considered in this
paper) can always be truthfully approximated via a simple VCG
mechanism. Finally, we make use of our technique to prove a lower
bound for the envy-minimization problem.

\subsection{The Setting}
In this setting we have  $m$ indivisible items and  $n$ strategic players. 
Each player $i$ is defined by a
private valuation function $v_i:2^{[m]}\rightarrow \mathbb{R_{+}}$. 
We assume that $v_i(\emptyset)=0$ (free disposal),  and $v_i(S)\le v_i(T)$ (monotonicity) for every  $i$ and every 
two sets of items $S,T\subseteq [m]$ such that $S\subseteq T.$ 

Each player  wishes to maximize the value of the 
set of items  assigned to it minus its payment to the mechanism.
If the mechanism is truthful then player $i$ can never improve its utility
by misreporting his private valuation function $v_i$ to the mechanism (no matter what the other players do). 
The set of alternatives $A$ contains all  possible  allocations of items to the players, where
all items must be assigned, and each item is assigned to exactly one player. 
The global goal of each of our three  notions of {\it non-utilitarian}
fairness -- Max-Min fairness, Min-Max fairness, and
envy-minimization -- will be defined below. 

\subsection{Max-Min Fairness}

\paragraph*{The global goal.} The Max-Min social choice function is concerned with
maximizing the value  of the least satisfied player. Formally, for
every $n$-tuple of $v_i$ valuations the Max-Min function assigns
the alternative $a\in A$ that maximizes the expression
 $\min_i \: v_i(a)$.

The Max-Min fairness allocation problem is 
based on the philosophical work of John Rawls~\cite{Rawls}.
Lavi et al.~\cite{LMN} proved that Max-Min fairness in allocations of
indivisible items cannot be \emph{optimally} implemented in a
truthful manner. 
Non-truthful algorithms for this
problem were designed~\cite{BS06}, as well as algorithms that settle
for restricted notions of truthfulness~\cite{BV05,Golo}. We prove
that no truthful deterministic mechanism can obtain \emph{any}
approximation ratio to the Max-Min fairness value. We prove this
lower bound even for the case of $2$ players and $2$ items.

\begin{theorem}
No truthful deterministic mechanism can obtain \emph{any}
approximation to the Max-Min fairness value in the allocation of
indivisible items. 
This holds even for the case of $2$ players and $2$ items.
\end{theorem}

\begin{proof}
Let $c > 1$ and let  $\epsilon$ be an arbitrarily small positive real number. 
Consider an instance with two players $i=1, 2$ and two
goods $g_a, g_b$. Each player $i$ has an additive valuation
function.  Let 
\[
\begin{array}{ll}
v_1(g_a)=2, & v_1(g_b)=\frac{1}{c} \\ 
v_2(g_a)=4-\epsilon, & v_2(g_b)=1+\epsilon.
\end{array} 
\]

 Note, that the optimal allocation assigns $g_a$
to player $1$ and $g_b$ to player $2$, thus obtaining a Max-Min value
of $1+\epsilon$. Also note, that this allocation will also be
chosen by any $c$-approximation mechanism.

We alter the valuation of player $2$ into $v'_2$ such that
$v'_2(g_a)=\frac{1}{c}, \ v'_2(g_b)=\frac{1}{c^2}-\epsilon$. 
The optimal Max-Min value is now $\frac{1}{c}$. Observe that any
$c$-approximation mechanism must assign $g_b$ to player $1$ and
$g_a$ to player $2$. However, if this happens we have that:

\[
(1+\epsilon)+\frac{1}{c}=v_2(g_b)+v'_2(g_a)<v_2(g_a)+v'_2(g_b)=(4-\epsilon)+\frac{1}{c^2} -\epsilon.
\]

This violates weak-monotonicity, and so no truthful
$c$-approximation mechanism exists. Since this is true for any $c > 1$
the theorem follows.
\end{proof}

\subsection{Min-Max Fairness}

\paragraph*{The global goal.}  Min-Max fairness can be thought of as the dual notion of Max-Min
fairness. It is relevant in settings in which each player incurs a
cost for every chosen alternative.~\footnote{Here each player wishes 
to minimize the cost of the 
set of items  assigned to it minus the payment made to it by the mechanism.}
 The Min-Max social choice
function is concerned with minimizing the cost incurred by the least
satisfied player. Formally, for every $n$-tuple of $v_i$ valuations
the Max-Min function assigns the alternative $a\in A$ that minimizes the
expression $\max_i \: v_i(a)$. 

Observe, that both the scheduling problem and the
workload-minimization problem discussed in this paper, are in fact
special cases of this notion of fairness. Studying Max-Min fairness
in this more abstract setting enables us to state this simple
observation -- any Min-Max social-choice function can be truthfully
approximated within a factor of $n$ (recall that $n$ is the number
of players) by a simple VCG mechanism. Since the best currently known
approximation-mechanisms for both scheduling and
workload-minimization are VCG-based, this result can be viewed as a
generalization of both.

\begin{theorem}
Let $f$ be a Min-Max social choice function. Then, there exists a
truthful deterministic mechanism that for every $n$-tuple of
valuations $v_1, ..., v_n$ outputs an alternative $a\in A$
 such that $\max_i\: v_i(a)$ is an $n$-approximation to the value of the solution $f$
outputs for these valuations.
\end{theorem}
\begin{proof}
Let $v_1, ..., v_n$ be the valuation function of the players. 
Let $b$ be the allocation that $f$
outputs for $v_1, ...,v_n$. Consider the VCG mechanism that minimizes
the total cost the players incur. The truthfulness of this mechanism
is guaranteed by the VCG technique. Let $a\in A$ be the allocation that this
mechanism outputs.
The result now follows
immediately from 
$\max_i \: v_i(a)  \le \Sigma_i \ v_i(a) $ and the fact that 
 $\max_i \: v_i(b)  \ge \frac{1}{n}\cdot  \Sigma_i \ v_i(a)$.

\end{proof}

\subsection{Envy-Minimization}

\paragraph*{The global goal.}  Lipton, Markakis, Mossel, and Saberi~\cite{LMSS} presented the
problem of envy-minimization for indivisible items.  They consider 
truthful mechanisms for this problem. 
An {\it envy-minimizing allocation} of items is a partition of the $m$
items into disjoint sets $S_1, ..., S_n$ (player $i$ is assigned $S_i$)
that minimizes the expression $\max_{i,j}\ \{ v_i(S_j) - v_i(S_i), 0 \}$ (over
all possible allocations). Intuitively, we wish to minimize the
maximal envy a player might feel by comparing his value for a set
of items given to another player to the value he assigns the items
allocated to him.  

They prove several approximability results
for this problem. The parameter considered in~\cite{LMSS} is the
\emph{maximal marginal utility}.

\begin{definition}
The $\mathrm{maximal \; marginal \; utility}$ $\alpha$ is defined as follows:
\[
\alpha(v_1, ..., v_n) =\max_{i\in [n],\  j\in [m],\  S\subseteq [m]}\ v_i(S\cup
\{j\})-v_i(S).
\]
\end{definition}

That is, $\alpha$ is the maximal value by which the value of a
player increases when one item is added to his bundle. 

Lipton et al.~\cite{LMSS}  exhibit a universally-truthful randomized mechanism that has an  envy of at most  
O$(\sqrt{\alpha}n^{\frac{1}{2}+\epsilon})$ w.h.p. for large values
of $n$. They show that no truthful mechanism can guarantee an
optimal solution with respect to  envy minimization. 
We strengthen this lower bound by
showing that no truthful deterministic mechanism can guarantee an
allocation that has an envy value within $\alpha$ from optimal.

\begin{theorem}\label{envy-thm}
Any truthful deterministic  mechanism cannot obtain an approximation ratio better
than $\alpha$ for envy minimization.
\end{theorem}

\begin{proof}
Assume for contradiction that there exists  a truthful deterministic mechanism $M$ 
that $\alpha$-approximates the envy minimization problem. 
Consider an instance with $2$ players and $3$ items. Each player
$i=1, 2$ has the same additive valuation function $v_i$ that assigns
any of the single items a value of $1$. Observe, that 
$\alpha(v_1, v_2) =1$. Notice, that the minimal envy 
for this instance is $1$.~\footnote{The {\it envy} of $S_1, ..., S_n$ equals 
 $\max_{i,j}\ \{ v_i(S_j) - v_i(S_i), 0 \}$.} 
Hence, if $M$ assigns all items to one of the players the
envy is precisely $3 > 1= \alpha(v_1, v_2) \cdot 1$. 
Therefore, we can assume w.l.o.g. that the 
 $\alpha$-approximation mechanism $M$
allocates items $1, 2$ to player $1$ 
and  item $3$ to player $2$ (call this partition $(S_1, S_2)$). 

Let  $\epsilon$ be an arbitrarily small positive real number.  We now
change the valuation function of player $1$ into the following
additive valuation:

\[
v'_1(j)=\left\{%
\begin{array}{ll}
1+\epsilon & j=1, 2 \\
\epsilon  & j=3. \\
\end{array}%
\right.
\]

Now $\alpha(v'_1, v_2) = 1+\epsilon$. Also observe that the minimal
envy for this new instance is $\epsilon$ (e.g., assign item $1$ to player $1$
and items $2, 3$ to player $2$). However, it is easy to 
verify that  weak-monotonicity dictates that the allocation
remains the same even after the alteration of the valuation of player
$1$ (since for any partition $(T_1, T_2)\neq (S_1, S_2)$ of the items 
we have that $v'_1(T_1)- v_1 (T_1) < v'_1(S_1) - v_1 (S_1) = 2\epsilon$).
Therefore, we end up with an allocation in which the envy is
 $1 > (1+\epsilon)\cdot \epsilon=\alpha(v'_1, v_2)\cdot \epsilon$, a contradiction.  
\end{proof}

\section*{Acknowledgements}
We thank   Ron Lavi, Noam Nisan, Chaitanya Swamy, Amir Ronen and anonymous referees for
helpful discussions and suggestions.

\bibliography{bib}

\begin{thebibliography}{10}

\bibitem{Azar}
Nir Andelman, Yossi Azar, and Motti Sorani.
\newblock Truthful approximation mechanisms for scheduling selfish related
  machines.
\newblock {\em Theory Comput. Syst.}, 40(4):423--436, 2007.

\bibitem{AT}
Aaron Archer and Eva Tardos.
\newblock Truthful mechanisms for one-parameter agents.
\newblock In {\em {IEEE} Symposium on Foundations of Computer Science}, pages
  482--491, 2001.

\bibitem{AT02}
Aaron Archer and {\'{E}}va Tardos.
\newblock Frugal path mechanisms.
\newblock {\em {ACM} Trans. Algorithms}, 3(1):3:1--3:22, 2007.

\bibitem{Ashlagi}
Itai Ashlagi, Shahar Dobzinski, and Ron Lavi.
\newblock Optimal lower bounds for anonymous scheduling mechanisms.
\newblock {\em Mathematics of Operations Research}, 37(2):244--258, 2012.

\bibitem{BS06}
Nikhil Bansal and Maxim Sviridenko.
\newblock The santa claus problem.
\newblock In {\em Proceedings of the 38th Annual {ACM} Symposium on Theory of
  Computing, Seattle, WA, USA, May 21-23, 2006}, pages 31--40, 2006.

\bibitem{BGN}
Yair Bartal, Rica Gonen, and Noam Nisan.
\newblock Incentive compatible multi unit combinatorial auctions.
\newblock In {\em Proceedings of the 9th Conference on Theoretical Aspects of
  Rationality and Knowledge (TARK-2003), Bloomington, Indiana, USA, June 20-22,
  2003}, pages 72--87, 2003.

\bibitem{BV05}
Ivona Bez{\'a}kov{\'a} and Varsha Dani.
\newblock Allocating indivisible goods.
\newblock {\em SIGecom Exchanges}, 5(3):11--18, 2005.

\bibitem{LMNB}
Sushil Bikhchandani, Shurojit Chatterji, Ron Lavi, Ahuva Mu'alem, Noam Nisan,
  and Arunava Sen.
\newblock Weak monotonicity characterizes deterministic dominant strategy
  implementation.
\newblock {\em Econometrica}, 74(4):1109--1132, July 2006.

\bibitem{Chawla-BIC-STOC13}
Shuchi Chawla, Jason~D. Hartline, David~L. Malec, and Balasubramanian Sivan.
\newblock Prior-independent mechanisms for scheduling.
\newblock In {\em Symposium on Theory of Computing Conference, {STOC} 2013,
  Palo Alto, CA, USA, June 1-4, 2013}, pages 51--60, 2013.

\bibitem{Chen}
Xujin Chen, Donglei Du, and Luis~Fernando Zuluaga.
\newblock Copula-based randomized mechanisms for truthful scheduling on two
  unrelated machines.
\newblock {\em Theory Comput. Syst.}, 57(3):753--781, 2015.

\bibitem{Fractional-sched}
George Christodoulou, Elias Koutsoupias, and Annam{\'{a}}ria Kov{\'{a}}cs.
\newblock Mechanism design for fractional scheduling on unrelated machines.
\newblock {\em {ACM} Trans. Algorithms}, 6(2):38:1--38:18, 2010.

\bibitem{CKV-carac-08}
George Christodoulou, Elias Koutsoupias, and Angelina Vidali.
\newblock A characterization of 2-player mechanisms for scheduling.
\newblock In {\em Algorithms - {ESA} 2008, 16th Annual European Symposium,
  Karlsruhe, Germany, September 15-17, 2008. Proceedings}, pages 297--307,
  2008.

\bibitem{KoutsoupiasSODA2007}
George Christodoulou, Elias Koutsoupias, and Angelina Vidali.
\newblock A lower bound for scheduling mechanisms.
\newblock {\em Algorithmica}, 55(4):729--740, 2009.

\bibitem{CK}
George Christodoulou and Annam{\'{a}}ria Kov{\'{a}}cs.
\newblock A deterministic truthful {PTAS} for scheduling related machines.
\newblock {\em {SIAM} J. Comput.}, 42(4):1572--1595, 2013.

\bibitem{Clarke}
Edward~H. Clarke.
\newblock Multipart pricing of public goods.
\newblock {\em Public Choice}, 11:17--33, 1971.

\bibitem{CA-sur}
Peter Cramton, Yoav Shoham, and Richard~Steinberg (eds.).
\newblock {\em Combinatorial Auctions}.
\newblock MIT Press, 2006.

\bibitem{MIKE15}
Amit Daniely, Michael Schapira, and Gal Shahaf.
\newblock Inapproximability of truthful mechanisms via generalizations of the
  {VC} dimension.
\newblock In {\em Proceedings of the Forty-Seventh Annual {ACM} on Symposium on
  Theory of Computing, {STOC} 2015, Portland, OR, USA, June 14-17, 2015}, pages
  401--408, 2015.

\bibitem{Costis-2015-sched}
Constantinos Daskalakis and S.~Matthew Weinberg.
\newblock Bayesian truthful \emph{Mechanisms} for job scheduling from
  bi-criterion approximation \emph{Algorithms}.
\newblock In {\em Proceedings of the Twenty-Sixth Annual {ACM-SIAM} Symposium
  on Discrete Algorithms, {SODA} 2015, San Diego, CA, USA, January 4-6, 2015},
  pages 1934--1952, 2015.

\bibitem{Gale-Demange-1985}
Gabrielle Demange and David Gale.
\newblock The strategy structure of two-sided matching markets.
\newblock {\em Econometrica}, 53(4):873--88, July 1985.

\bibitem{sched-rfpas}
Peerapong Dhangwatnotai, Shahar Dobzinski, Shaddin Dughmi, and Tim Roughgarden.
\newblock Truthful approximation schemes for single-parameter agents.
\newblock {\em {SIAM} J. Comput.}, 40(3):915--933, 2011.

\bibitem{DD13}
Shahar Dobzinski and Shaddin Dughmi.
\newblock On the power of randomization in algorithmic mechanism design.
\newblock {\em {SIAM} J. Comput.}, 42(6):2287--2304, 2013.

\bibitem{DN10}
Shahar Dobzinski and Noam Nisan.
\newblock Mechanisms for multi-unit auctions.
\newblock {\em J. Artif. Intell. Res. {(JAIR)}}, 37:85--98, 2010.

\bibitem{DNS05}
Shahar Dobzinski, Noam Nisan, and Michael Schapira.
\newblock Approximation algorithms for combinatorial auctions with
  complement-free bidders.
\newblock {\em Math. Oper. Res.}, 35(1):1--13, 2010.

\bibitem{DNS06}
Shahar Dobzinski, Noam Nisan, and Michael Schapira.
\newblock Truthful randomized mechanisms for combinatorial auctions.
\newblock {\em J. Comput. Syst. Sci.}, 78(1):15--25, 2012.

\bibitem{Dobzin-mukund-charac-08}
Shahar Dobzinski and Mukund Sundararajan.
\newblock On characterizations of truthful mechanisms for combinatorial
  auctions and scheduling.
\newblock In {\em Proceedings 9th {ACM} Conference on Electronic Commerce
  (EC-2008), Chicago, IL, USA, June 8-12, 2008}, pages 38--47, 2008.

\bibitem{Dob16}
Shahar Dobzinski and Jan Vondr{\'{a}}k.
\newblock Impossibility results for truthful combinatorial auctions with
  submodular valuations.
\newblock {\em J. {ACM}}, 63(1):5:1--5:19, 2016.

\bibitem{FKMS}
Joan Feigenbaum, David~R. Karger, Vahab~S. Mirrokni, and Rahul Sami.
\newblock Subjective-cost policy routing.
\newblock {\em Theor. Comput. Sci.}, 378(2):175--189, 2007.

\bibitem{FPSS}
Joan Feigenbaum, Christos Papadimitriou, Rahul Sami, and Scott Shenker.
\newblock A $\mathrm{BGP}$-based mechanism for lowest-cost routing.
\newblock {\em Distributed Computing}, 18:61--72, 2005.

\bibitem{FRS}
Joan Feigenbaum, Vijay Ramachandran, and Michael Schapira.
\newblock Incentive-compatible interdomain routing.
\newblock {\em Distributed Computing}, 23(5-6):301--319, 2011.

\bibitem{FSS}
Joan Feigenbaum, Rahul Sami, and Scott Shenker.
\newblock Mechanism design for policy routing.
\newblock {\em Distributed Computing}, 18:293--305, 2006.

\bibitem{FGHK02}
Amos Fiat, Andrew~V. Goldberg, Jason~D. Hartline, and Anna~R. Karlin.
\newblock Competitive generalized auctions.
\newblock In {\em Proceedings on 34th Annual {ACM} Symposium on Theory of
  Computing, May 19-21, 2002, Montr{\'{e}}al, Qu{\'{e}}bec, Canada}, pages
  72--81, 2002.

\bibitem{Gamzu}
Iftah Gamzu.
\newblock Improved lower bounds for non-utilitarian truthfulness.
\newblock {\em Theor. Comput. Sci.}, 412(7):626--632, 2011.

\bibitem{Maria-BIC-WINE15}
Yiannis Giannakopoulos and Maria Kyropoulou.
\newblock The {VCG} mechanism for bayesian scheduling.
\newblock In {\em Web and Internet Economics - 11th International Conference,
  {WINE} 2015, Amsterdam, The Netherlands, December 9-12, 2015, Proceedings},
  pages 343--356, 2015.

\bibitem{Golo}
Daniel Golovin.
\newblock Max-min fair allocation of indivisible goods, 2005.
\newblock Technical Report CMU-CS-05-144, Carnegie Mellon University.

\bibitem{Gro73}
Theodore Groves.
\newblock Incentives in teams.
\newblock {\em Econometrica}, 41(4):617--631, 1973.

\bibitem{HKMT}
Ron Holzman, Noa Kfir-Dahav, Dov Monderer, and Moshe Tennenholtz.
\newblock Bundling equilibrium in combinatorial auctions.
\newblock {\em Games and Economic Behavior}, 47:104--123, 2004.

\bibitem{HS76}
Ellis Horowitz and Sartaj Sahni.
\newblock Exact and approximate algorithms for scheduling nonidentical
  processors.
\newblock {\em J. {ACM}}, 23(2):317--327, 1976.

\bibitem{KKT}
Anna~R. Karlin, David Kempe, and Tami Tamir.
\newblock Beyond {VCG:} frugality of truthful mechanisms.
\newblock In {\em 46th Annual {IEEE} Symposium on Foundations of Computer
  Science {(FOCS} 2005), 23-25 October 2005, Pittsburgh, PA, USA, Proceedings},
  pages 615--626, 2005.

\bibitem{KoutsoupiasPHI}
Elias Koutsoupias and Angelina Vidali.
\newblock A lower bound of 1+\emph{{\(\varphi\)}} for truthful scheduling
  mechanisms.
\newblock {\em Algorithmica}, 66(1):211--223, 2013.

\bibitem{KV-smon-15}
Annam{\'{a}}ria Kov{\'{a}}cs and Angelina Vidali.
\newblock A characterization of n-player strongly monotone scheduling
  mechanisms.
\newblock In {\em Proceedings of the Twenty-Fourth International Joint
  Conference on Artificial Intelligence, {IJCAI} 2015, Buenos Aires, Argentina,
  July 25-31, 2015}, pages 568--574, 2015.

\bibitem{LMN}
Ron Lavi, Ahuva Mu'alem, and Noam Nisan.
\newblock Towards a characterization of truthful combinatorial auctions.
\newblock In {\em 44th Symposium on Foundations of Computer Science {(FOCS}
  2003), 11-14 October 2003, Cambridge, MA, USA, Proceedings}, pages 574--583,
  2003.

\bibitem{Lavi-Swamy-sched}
Ron Lavi and Chaitanya Swamy.
\newblock {Truthful mechanism design for multidimensional scheduling via cycle
  monotonicity}.
\newblock {\em Games and Economic Behavior}, 67(1):99--124, 2009.
\newblock Special Section of Games and Economic Behavior Dedicated to the 8th
  ACM Conference on Electronic Commerce.

\bibitem{LS05}
Ron Lavi and Chaitanya Swamy.
\newblock Truthful and near-optimal mechanism design via linear programming.
\newblock {\em J. {ACM}}, 58(6):25:1--25:24, 2011.

\bibitem{LOS}
Daniel Lehmann, Liadan O'Callaghan, and Yoav Shoham.
\newblock Truth revelation in approximately efficient combinatorial auctions.
\newblock {\em Journal of the ACM}, 49(5):577--602, 2002.

\bibitem{LST87}
Jan~Karel Lenstra, David~B. Shmoys, and {\'{E}}va Tardos.
\newblock Approximation algorithms for scheduling unrelated parallel machines.
\newblock {\em Math. Program.}, 46:259--271, 1990.

\bibitem{LMSS}
Richard~J. Lipton, Evangelos Markakis, Elchanan Mossel, and Amin Saberi.
\newblock On approximately fair allocations of indivisible goods.
\newblock In {\em Proceedings 5th {ACM} Conference on Electronic Commerce
  (EC-2004), New York, NY, USA, May 17-20, 2004}, pages 125--131, 2004.

\bibitem{Lu2009}
Pinyan Lu.
\newblock On 2-player randomized mechanisms for scheduling.
\newblock In {\em Internet and Network Economics, 5th International Workshop,
  {WINE} 2009, Rome, Italy, December 14-18, 2009. Proceedings}, pages 30--41,
  2009.

\bibitem{LY}
Pinyan Lu and Changyuan Yu.
\newblock An improved randomized truthful mechanism for scheduling unrelated
  machines.
\newblock In {\em {STACS} 2008, 25th Annual Symposium on Theoretical Aspects of
  Computer Science, Bordeaux, France, February 21-23, 2008, Proceedings}, pages
  527--538, 2008.

\bibitem{Mas}
Andreu Mas-Collel, Michael~D. Whinston, and Jerry~R. Green.
\newblock {\em Microeconomic Theory}.
\newblock Oxford university press, 1995.

\bibitem{Mu'alem-FAIR-by-Design}
Ahuva Mu'alem.
\newblock Fair by design: Multidimensional envy-free mechanisms.
\newblock {\em Games and Economic Behavior}, 88:29--46, 2014.

\bibitem{MN02}
Ahuva Mu'alem and Noam Nisan.
\newblock Truthful approximation mechanisms for restricted combinatorial
  auctions.
\newblock {\em Games and Economic Behavior}, 64:612--631, 2008.

\bibitem{MS08}
Ahuva Mu'alem and Michael Schapira.
\newblock Mechanism design over discrete domains.
\newblock In {\em Proceedings 9th {ACM} Conference on Electronic Commerce
  (EC-2008), Chicago, IL, USA, June 8-12, 2008}, pages 31--37, 2008.

\bibitem{Nisan-MU-Survey}
Noam Nisan.
\newblock Chapter 9 - algorithmic mechanism design: Through the lens of
  multiunit auctions.
\newblock volume~4 of {\em Handbook of Game Theory with Economic Applications},
  pages 477--515. Elsevier, 2015.

\bibitem{NR}
Noam Nisan and Amir Ronen.
\newblock Algorithmic mechanism design.
\newblock {\em Games and Economic Behavior}, 35:166--196, 2001.

\bibitem{AGT-Book}
Noam Nisan, Tim Roughgarden, Eva Tardos, and Vijay V.~Vazirani (eds.).
\newblock {\em Algorithmic Game Theory}.
\newblock Cambridge University Press, 2007.

\bibitem{NS02}
Noam Nisan and Ilya Segal.
\newblock The communication requirements of efficient allocations and
  supporting prices.
\newblock {\em J. Economic Theory}, 129(1):192--224, 2006.

\bibitem{Rubi}
Martin~J. Osborne and Ariel Rubinstein.
\newblock {\em A Course in Game Theory}.
\newblock MIT press, 1994.

\bibitem{Rawls}
John Rawls.
\newblock {\em A Theory of Justice}.
\newblock Cambridge, MA: Belknap Press of Harvard University Press, 1971.

\bibitem{Tim-Book-2016}
Tim Roughgarden.
\newblock {\em Twenty Lectures on Algorithmic Game Theory}.
\newblock Cambridge University Press, 2016.

\bibitem{Stein}
Hugo Steinhaus.
\newblock The problem of fair division.
\newblock {\em Econometrica}, 16(1):101--104, 1948.

\bibitem{Vic61}
William Vickrey.
\newblock Counterspeculation, auctions and competitive sealed tenders.
\newblock {\em Journal of Finance}, 16(1):8--37, 1961.

\bibitem{Vidali2009}
Angelina Vidali.
\newblock The geometry of truthfulness.
\newblock In {\em Internet and Network Economics, 5th International Workshop,
  {WINE} 2009, Rome, Italy, December 14-18, 2009. Proceedings}, pages 340--350,
  2009.

\bibitem{Vidali2011}
Angelina Vidali.
\newblock Extending characterizations of truthful mechanisms from subdomains to
  domains.
\newblock In {\em Internet and Network Economics - 7th International Workshop,
  {WINE} 2011, Singapore, December 11-14, 2011. Proceedings}, pages 408--414,
  2011.

\bibitem{Yao}
Andrew~Chi{-}Chih Yao.
\newblock Probabilistic computations: Toward a unified measure of complexity
  (extended abstract).
\newblock In {\em 18th Annual Symposium on Foundations of Computer Science
  {FOCS}, Providence, Rhode Island, USA, 31 October - 1 November 1977}, pages
  222--227, 1977.

\bibitem{Yu09}
Changyuan Yu.
\newblock Truthful mechanisms for two-range-values variant of unrelated
  scheduling.
\newblock {\em Theor. Comput. Sci.}, 410(21-23):2196--2206, 2009.

\end{thebibliography}

\appendix

\section{Appendix}
\subsection{A Universally-Truthful Randomized Approximation Mechanism for the
Scheduling Problem with Unrelated Machines} \label{upper-bound}

Nisan and Ronen~\cite{NR} present a truthful deterministic mechanism
that obtains an $n$-approximation. For the case of $2$ machines,
they exhibit a universally-truthful randomized mechanism that
obtains an approximation of $\frac{7}{4}$ (this bound was later improved  to 1.58606 
by~\cite{Chen}).  
We generalize this result
by presenting a universally-truthful randomized mechanism that
obtains an approximation-ratio of $\frac{7n}{8}$.
We now turn to the description of our mechanism for $n$ machines:

\vspace{0.1in} \noindent \textbf{Input:} An $n$-tuple of valuations $(v_1, ..., v_n)$ and $m$ tasks.

\vspace{0.1in} \noindent \textbf{Output:} An allocation
$T=T_1$, ..., $T_n$ of tasks, and payments $p_1, ..., p_n$ such that $T$
has a makespan value which is a $\frac{7n}{8}$-approximation to the
optimal makespan value, and the payments induce truthfulness.

\vspace{0.1in} \noindent {\bf The Mechanism:} \vspace{0.1in}
\begin{enumerate}
\item For every machine $i$ let $T_i\leftarrow \emptyset$ and $p_i\leftarrow 0$.
\item Partition the set of machines into two sets $S_1\leftarrow \{1, ..., \frac{n}{2}\}$ and
$S_2\leftarrow\{\frac{n}{2}+1, ..., n\}$.
\item For each task
$j=1, ..., m$ perform the following actions (assuming that $\min$ and $\mbox{argmin}$ 
break ties arbitrarily):

\begin{itemize}
\item Let $v^1\leftarrow\min_{i\in S_1}\ v_i(j)$, and let $i^1\leftarrow\mbox{argmin}_{i\in S_1}\ v_i(j)$.
\item Let $v'^1\leftarrow\min_{i\in S_1\setminus \{i^1\}}\ v_i(j)$.
\item Let $v^2\leftarrow\min_{i\in S_2}\ v_i(j)$,  and let $i^2\leftarrow\mbox{argmin}_{i\in S_2}\ v_i(j)$.
\item Let $v'^2\leftarrow\min_{i\in S_2\setminus \{i^2\}}\ v_i(j)$.
\item Randomly and uniformly choose a value $R\in \{0, 1\}$.
\item If $R=0$ and $v^1\le \frac{4}{3}v^2$ 
set $T_{i^1}\leftarrow T_{i^1} \bigcup\; \{ j \}$ and set $p_{i^1}\leftarrow p_{i^1}+\min\{v'^1,\frac{4}{3}v^2\}$.
\item If $R=0$ and $v^1> \frac{4}{3}v^2$ 
set $T_{i^2}\leftarrow T_{i^2}\bigcup\; \{ j \}$ and set $p_{i^2}\leftarrow p_{i^2}+\min\{v'^2,\frac{3}{4}v^1\}$.
\item If $R=1$ and $v^2\le \frac{4}{3}v^1$ 
set $T_{i^2}\leftarrow T_{i^2}\bigcup\; \{ j \}$ and set $p_{i^2}\leftarrow p_{i^2}+\min\{v'^2,\frac{4}{3}v^1\}$.
\item If $R=1$ and $v^2> \frac{4}{3}v^1$ 
set $T_{i^1}\leftarrow T_{i^1}\bigcup\;\{ j \}$ and set $p_{i^1}\leftarrow p_{i^1}+\min\{v'^1,\frac{3}{4}v^2\}$.
\end{itemize}
\item Allocate each machine $i$ the tasks in $T_i$, and pay it a sum of $p_i$.
\end{enumerate}

\begin{remark}
If $n$ cannot be divided by $2$ simply add the extra machine to either $S_1$ or $S_2$.
\end{remark}

\begin{theorem}
There exists a universally-truthful randomized mechanism for the
scheduling problem that obtains an approximation ratio of $\frac{7n}{8}$.
\end{theorem}
\begin{proof}
We prove the theorem for the case that $n$ can be divided by $2$.
The proof for the other case is similar. Our proof relies on the
proof of Nisan and Ronen~\cite{NR}. Observe, that the utility of
each machine after the algorithm finishes is the sum of its
utilities for the different tasks. Hence, it is sufficient to prove
that for each individual task a machine has no incentive to lie. As
in~\cite{NR}, this is guaranteed because the allocation of each task
is in fact a weighted VCG mechanism (see~\cite{NR} for further
explanations), which is known to be truthful. Hence, this mechanism
is universally truthful.

We now need to prove that the approximation ratio guaranteed by the
mechanism is indeed  $\frac{7n}{8}$. Let $\mathcal{I}$ be an instance of the
scheduling problem with $m$ tasks, and with $n$ machines that have
the valuation functions $v_1, ..., v_n$. We define an instance $\mathcal{I}'$ of
scheduling problem with $m$ tasks, and with $2$ machines that have
the valuation function $v'_1,v'_2$, in the following way:  $v'_1(j)=min_{i\in S_1}\ v_i(j)$ for all
$j\in [m]$. Similarly,  $v'_2(j)=min_{i\in S_2}\ v_i(j)$ for all
$j\in [m]$. We denote by $M(\mathcal{I})$ and
by $M(\mathcal{I}')$ the makespan values our mechanism generates for $\mathcal{I}$ and
$\mathcal{I}'$ respectively. We denote by $OPT(\mathcal{I})$ and by $OPT(\mathcal{I}')$ the optimal
makespan values for $\mathcal{I}$ and $\mathcal{I}'$ respectively.

First, notice that $M(\mathcal{I})\le M(\mathcal{I}')$. This is because applying our
mechanism to $\mathcal{I}'$ results in the same makespan value as applying it
to $\mathcal{I}$ in the worst-case scenario in which tasks are always assigned
to the same machines in $S_1$ and in $S_2$. It also holds that
$M(\mathcal{I}')\le \frac{7}{4}OPT(\mathcal{I}')$ because in the case that there are only
two machines our mechanism is precisely that of~\cite{NR}, which
guarantees a $\frac{7}{4}$ approximation ratio. We now have that
$M(\mathcal{I})\le \frac{7}{4}OPT(\mathcal{I}')$. All that is left to show is that
$OPT(\mathcal{I}')\le \frac{n}{2}OPT(\mathcal{I})$. Consider the optimal allocation of tasks
for $\mathcal{I}$. By giving all tasks assigned to machines in $S_1$ to
machine $1$ in $\mathcal{I}'$, and allocating all tasks assigned to machines in
$S_2$ to machine $2$ in $\mathcal{I}'$, we end up with a makespan value for $\mathcal{I}'$
that is at most $\frac{n}{2}OPT(\mathcal{I})$. The theorem follows.
\end{proof}

\subsection{Lower Bounds for  minimizing the workload in
inter-domain routing}\label{LBs-Inter-Domain}

\begin{theorem}\label{thm-work}
No truthful deterministic mechanism for minimizing the workload in
inter-domain routing can obtain an approximation ratio better than
$\frac{1 + \sqrt{5}}{2} \approx 1.618$.
\end{theorem}

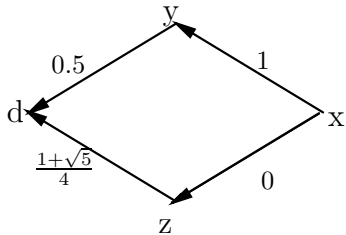
\begin{figure}[hbt]
\begin{center}
\setlength{\unitlength}{0.00052493in}
\begingroup\makeatletter\ifx\SetFigFont\undefined%
\gdef\SetFigFont#1#2#3#4#5{%
  \reset@font\fontsize{#1}{#2pt}%
  \fontfamily{#3}\fontseries{#4}\fontshape{#5}%
  \selectfont}%
\fi\endgroup%
{\renewcommand{\dashlinestretch}{30}
\begin{picture}(3393,2460)(0,-10)
\path(3169,1223)(1684,323)
\blacken\thicklines
\path(1827.440,462.553)(1684.000,323.000)(1874.087,385.585)(1827.440,462.553)
\path(3165,1230)(1680,330)
\blacken\thinlines
\path(1767.075,417.852)(1680.000,330.000)(1798.173,366.540)(1767.075,417.852)
\thicklines
\path(1702,351)(217,1251)
\blacken\path(407.087,1188.415)(217.000,1251.000)(360.440,1111.447)(407.087,1188.415)
\path(3206,1238)(1721,2138)
\blacken\path(1911.087,2075.415)(1721.000,2138.000)(1864.440,1998.447)(1911.087,2075.415)
\path(1722,2134)(237,1234)
\blacken\path(380.440,1373.553)(237.000,1234.000)(427.087,1296.585)(380.440,1373.553)
\put(2535,1680){\makebox(0,0)[lb]{\smash{{\SetFigFont{10}{12.0}{\rmdefault}{\mddefault}{\updefault}1}}}}
\put(465,1635){\makebox(0,0)[lb]{\smash{{\SetFigFont{10}{12.0}{\rmdefault}{\mddefault}{\updefault}0.5}}}}
\put(1545,15){\makebox(0,0)[lb]{\smash{{\SetFigFont{12}{14.4}{\rmdefault}{\mddefault}{\updefault}z}}}}
\put(1590,2175){\makebox(0,0)[lb]{\smash{{\SetFigFont{12}{14.4}{\rmdefault}{\mddefault}{\updefault}y}}}}
\put(3255,1095){\makebox(0,0)[lb]{\smash{{\SetFigFont{12}{14.4}{\rmdefault}{\mddefault}{\updefault}x}}}}
\put(15,1140){\makebox(0,0)[lb]{\smash{{\SetFigFont{12}{14.4}{\rmdefault}{\mddefault}{\updefault}d}}}}
\put(2580,465){\makebox(0,0)[lb]{\smash{{\SetFigFont{10}{12.0}{\rmdefault}{\mddefault}{\updefault}0}}}}
\put(285,600){\makebox(0,0)[lb]{\smash{{\SetFigFont{10}{12.0}{\rmdefault}{\mddefault}{\updefault}$\frac{1 + \sqrt 5}{4}$}}}}
\end{picture}
} 
\caption{The Instance $\mathcal{I}$} \label{fig2}
\end{center}
\end{figure}

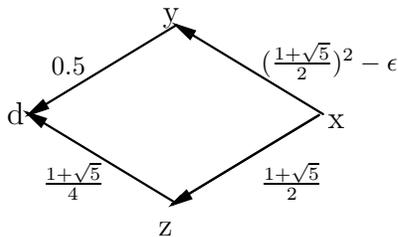
\begin{figure}[hbt]
\begin{center}
\setlength{\unitlength}{0.00052493in}
\begingroup\makeatletter\ifx\SetFigFont\undefined%
\gdef\SetFigFont#1#2#3#4#5{%
  \reset@font\fontsize{#1}{#2pt}%
  \fontfamily{#3}\fontseries{#4}\fontshape{#5}%
  \selectfont}%
\fi\endgroup%
{\renewcommand{\dashlinestretch}{30}
\begin{picture}(3393,2460)(0,-10)
\path(3169,1223)(1684,323)
\blacken\thicklines
\path(1827.440,462.553)(1684.000,323.000)(1874.087,385.585)(1827.440,462.553)
\path(3165,1230)(1680,330)
\blacken\thinlines
\path(1767.075,417.852)(1680.000,330.000)(1798.173,366.540)(1767.075,417.852)
\thicklines
\path(1702,351)(217,1251)
\blacken\path(407.087,1188.415)(217.000,1251.000)(360.440,1111.447)(407.087,1188.415)
\path(3206,1238)(1721,2138)
\blacken\path(1911.087,2075.415)(1721.000,2138.000)(1864.440,1998.447)(1911.087,2075.415)
\path(1722,2134)(237,1234)
\blacken\path(380.440,1373.553)(237.000,1234.000)(427.087,1296.585)(380.440,1373.553)
\put(465,1635){\makebox(0,0)[lb]{\smash{{\SetFigFont{10}{12.0}{\rmdefault}{\mddefault}{\updefault}0.5}}}}
\put(1545,15){\makebox(0,0)[lb]{\smash{{\SetFigFont{12}{14.4}{\rmdefault}{\mddefault}{\updefault}z}}}}
\put(1590,2175){\makebox(0,0)[lb]{\smash{{\SetFigFont{12}{14.4}{\rmdefault}{\mddefault}{\updefault}y}}}}
\put(3255,1095){\makebox(0,0)[lb]{\smash{{\SetFigFont{12}{14.4}{\rmdefault}{\mddefault}{\updefault}x}}}}
\put(15,1140){\makebox(0,0)[lb]{\smash{{\SetFigFont{12}{14.4}{\rmdefault}{\mddefault}{\updefault}d}}}}
\put(2580,1680){\makebox(0,0)[lb]{\smash{{\SetFigFont{10}{12.0}{\rmdefault}{\mddefault}{\updefault}$(\frac{1+\sqrt 5 }{2}) ^2 -\epsilon$}}}}
\put(2580,465){\makebox(0,0)[lb]{\smash{{\SetFigFont{10}{12.0}{\rmdefault}{\mddefault}{\updefault}$\frac{1 + \sqrt 5}{2}$}}}}
\put(375,465){\makebox(0,0)[lb]{\smash{{\SetFigFont{10}{12.0}{\rmdefault}{\mddefault}{\updefault}$\frac{1 + \sqrt 5}{4}$}}}}
\end{picture}
} 
\caption{The Instance $\mathcal{I}'$} \label{fig3}
\end{center}
\end{figure}

\begin{proof}
This proof is similar to the proof of
Theorem~\ref{deterministic-bound}. To prove the lower bound consider
the instances of the workload-minimization problem with $3$ source
nodes $x, y, z$ depicted in Figures $2$ and $3$. Each source node
has a single packet it wishes to send to the destination. The number
beside every directed link $(u, u')$ in these figures represents the
cost $u$ incurs for transferring a packet to $u'$. 
 Denote the instance in Figure $2$ by $\mathcal{I}$
and the instance in Figure $3$ by $\mathcal{I}'$. Observe that only the cost
function of node $x$ is different in $\mathcal{I}$ and $\mathcal{I}'$. We denote the
cost function of $x$ in $\mathcal{I}$ by $c_x$ and his cost function in
$\mathcal{I}'$ by $c'_x$.

Assume, by contradiction, that $M$ is a truthful deterministic
mechanism that obtains an approximation ratio better than
$\phi = \frac{1 + \sqrt{5}}{2}$. 
Observe, that for the instance $\mathcal{I}$ $M$ must direct the
traffic originating in node $x$ through node $y$ (otherwise this
contradicts the fact that $M$ obtains an approximation ratio better
than $\phi$). Similarly, for the instance $\mathcal{I}'$ $M$ must direct the
traffic originating in node $x$ through node $z$. However, this
violates the weak-monotonicity of $M$ as
$1+\phi =c_x((x, y))+c'_x((x, z))>c_x((x, y))+c'_x((x, z))=\phi^2 -\epsilon$ (since $\phi$ is the golden ratio).
\end{proof}

\vspace{0.1in}

\begin{theorem}
No universally-truthful randomized mechanism for minimizing the
workload in inter-domain routing can obtain an approximation ratio
better than $\frac{3+\sqrt{5}}{4}\approx 1.309$.
\end{theorem}
\begin{proof}
 We define $\mathcal{I}$ and $\mathcal{I}'$ as in
the proof of Theorem~\ref{thm-work}. Consider the uniform
distribution over $\mathcal{I}$ and $\mathcal{I}'$. Let $M$ be a truthful
deterministic mechanism. As shown in the proof of Theorem~\ref{thm-work}, 
$M$ cannot achieve an approximation better than
$\frac{1+\sqrt{5}}{2}$ on both $\mathcal{I}$ and $\mathcal{I}'$ due to its weak-monotonicity.
Therefore, the expected approximation of $M$ is at least
$\frac{1}{2}\times 1 + \frac{1}{2}\times \frac{1 + \sqrt{5}}{2} \approx 1.309$.
\end{proof}
\end{document}